\pdfoutput=1
\documentclass[12pt]{article}
\usepackage{amsmath}
\usepackage{graphicx}
\usepackage{natbib}
\usepackage{url} 

\RequirePackage{amsthm,amsmath,amsfonts,amssymb}
\RequirePackage{graphicx}

\setcitestyle{authoryear,open={(},close={)}}

\usepackage{tikz}
\usetikzlibrary{positioning}
\usepackage{cleveref}
\usepackage{caption}
\usepackage{subcaption}
\usepackage{algorithm}
\usepackage{algorithmic}
\usepackage{mathtools}

\usepackage{custom}

\newcommand{\blind}{0}

\begin{document}

\def\spacingset#1{\renewcommand{\baselinestretch}%
{#1}\small\normalsize} \spacingset{1}

\if0\blind
{
  \title{\bf Several Remarks on the Numerical Integrator in Lagrangian Monte Carlo}
  \author{James A. Brofos\thanks{
    The authors gratefully acknowledge the Yale Center for Research Computing for use of the research computing infrastructure. We thank Marcus A. Brubaker for helpful discussions.
This material is based upon work supported by the National Science Foundation Graduate Research Fellowship under Grant No. 1752134. Any opinion, findings, and conclusions or recommendations expressed in this material are those of the authors(s) and do not necessarily reflect the views of the National Science Foundation. The work is also supported in part by NIH/NIGMS 1R01GM136780-01 and AFSOR FA9550-21-1-0317.}\hspace{.2cm}\\
    Department of Statistics and Data Science, Yale University \\
    and \\
    Roy R. Lederman \\
    Department of Statistics and Data Science, Yale University }
  \maketitle
} \fi

\if1\blind
{
  \bigskip
  \bigskip
  \bigskip
  \begin{center}
    {\LARGE\bf Title}
\end{center}
  \medskip
} \fi

\bigskip
\begin{abstract}
  Riemannian manifold Hamiltonian Monte Carlo (RMHMC) is a powerful method of Bayesian inference that exploits underlying geometric information of the posterior distribution in order to efficiently traverse the parameter space. However, the form of the Hamiltonian necessitates complicated numerical integrators, such as the generalized leapfrog method, that preserve the detailed balance condition. The distinguishing feature of these numerical integrators is that they involve solutions to implicitly defined equations. Lagrangian Monte Carlo (LMC) proposes to eliminate the fixed point iterations by transitioning from the Hamiltonian formalism to Lagrangian dynamics, wherein a fully explicit integrator is available. This work makes several contributions regarding the numerical integrator used in LMC. First, it has been claimed in the literature that the integrator is only first-order accurate for the Lagrangian equations of motion; to the contrary, we show that the LMC integrator enjoys second order accuracy. Second, the current conception of LMC requires four determinant computations in every step in order to maintain detailed balance; we propose a simple modification to the integration procedure in LMC in order to reduce the number of determinant computations from four to two while still retaining a fully explicit numerical integration scheme. Third, we demonstrate that the LMC integrator enjoys a certain robustness to human error that is not shared with the generalized leapfrog integrator, which can invalidate detailed balance in the latter case. We discuss these contributions within the context of several benchmark Bayesian inference tasks.
\end{abstract}

\noindent%
\vfill

\newpage
\spacingset{1.5} 

\section{Introduction}

Let $\mathcal{L} : \R^m\to\R$ be the log-density of a smooth distribution, known up to an additive constant. A critical problem in Bayesian inference is the design of procedures that generate samples from a target density $\pi(q) \propto \exp(\mathcal{L}(q))$. Hamiltonian Monte Carlo (HMC) \citep{Duane1987216,betancourt2017conceptual} is a Markov chain Monte Carlo (MCMC) method for sampling from arbitrary differentiable probability distributions based on numerical solutions to Hamilton's equations of motion. Unlike random walk Metropolis or the Metropolis-adjusted Langevin algorithm, HMC can propose distant states, thereby dramatically decreasing the autocorrelation between states of the chain and increasing sampling efficiency.

The efficiency of HMC can be further improved by incorporating geometric concepts into the proposal mechanism. Information geometry \citep{10.5555/3019383} provides a framework for representing sets of probability densities as a Riemannian manifold whose metric can be chosen as the Fisher information matrix. When $\mathcal{L}$ decomposes into the sum of a log-likelihood and a log-prior terms, \citet{rmhmc} proposed Riemannian manifold Hamiltonian Monte Carlo (RMHMC), which sought to precondition Hamiltonian dynamics with the inverse of the sum of Fisher information of the log-likelihood and the negative Hessian of the log-prior. In general, both the Fisher information and the Hessian of the log-prior depend on the parameter, denoted here by $q$, to be sampled; this dependency produces a complicated Hamiltonian. In order to maintain the detailed balance condition of the RMHMC Markov chain, elaborate numerical integrators such as the generalized leapfrog or implicit midpoint methods must be used \citep{pourzanjani2019implicit,rmhmc,pmlr-v139-brofos21a}. These numerical integrators involve implicitly defined updates, wherein updates are defined as the solution to a fixed point equation, which are typically resolved to a prescribed convergence tolerance by fixed point iteration.

\citet{lan2015} introduced Lagrangian Monte Carlo (LMC) as an alternative to RMHMC. Unlike RMHMC, which is based on numerical solutions to Hamiltonian mechanics, LMC is instead inspired by the Lagrangian formalism of classical mechanics. Although Lagrangian and Hamiltonian mechanics are formally equivalent from a physical perspective, methods of numerical integration assume a simplified form when expressed as Lagrangian mechanics. Indeed, the simplification is so significant that \citet{lan2015} was able to devise a fully explicit numerical integrator for which detailed balance could be maintained when used as a proposal operator in MCMC. However, the elimination of fixed point iterations was replaced by the requirement that four Jacobian determinants be computed in a single step of the explicit integrator. Moreover, it was stated that the explicit method had only first-order accuracy as a numerical integrator, supposedly producing less accurate solutions than the generalized leapfrog method, whose accuracy is second-order.

The purpose of the present work is three-fold. First, we propose a simple mechanism by which to reduce the number of Jacobian determinant evaluations from four to two in a single step of the integrator. This is achieved by inverting the sequence in which position and velocity are integrated. Second, we will clarify that the order of the explicit integrator, with or without inversion, actually has second-order accuracy, the same as those integrators which are commonly used in RMHMC. Third, we discuss how LMC enjoys a greater robustness to human error than RMHMC. The outline of the remainder of this paper is as follows. In \cref{lagrangian-remarks:sec:preliminaries} we discuss preliminary material on Hamiltonian and Lagrangian mechanics, numerical integrators, and MCMC methods based off of these physical models. In \cref{lagrangian-remarks:sec:related-work} we review some related work in the literature on integration methods for RMHMC. In \cref{lagrangian-remarks:sec:analytical-apparatus} we proceed to our analysis of the LMC integrator, where we describe how one may reduce the number of Jacobian determinant computations while maintaining a fully explicit integration method and give a proof that the integrator has second-order accuracy. In \cref{lagrangian-remarks:sec:experimentation} we turn to the evaluation of the proposed modifications to LMC; we evaluate performance on several benchmark Bayesian inference tasks and give numerical evidence to support the claim of second-order accuracy.

\section{Preliminaries}\label{lagrangian-remarks:sec:preliminaries}

In this section we review the necessary background for Hamiltonian and Lagrangian Monte Carlo. \Cref{lagrangian-remarks:subsec:essentials} reviews the most important concepts from RMHMC and LMC, giving perspective on the varieties of integrators and how they are employed in a Markov chain Monte Carlo procedure. We then proceed in \cref{lagrangian-remarks:subsec:hamiltonian-lagrangian-mechanics} to recall the Hamiltonian and Lagrangian formalisms from physics. \Cref{lagrangian-remarks:subsec:numerical-integrators} then treats the matter of numerical integration of the Hamiltonian and Lagrangian mechanics. In \cref{lagrangian-remarks:subsec:hamiltonian-lagrangian-monte-carlo} we then review Hamiltonian and Lagrangian Monte Carlo using the framework of involutive Monte Carlo, wherein we give special attention to the Jacobian determinant computations that are necessitated in the Lagrangian construction.

\subsection{Notation}

In the context of classical mechanics, we denote by $q$ the position variable, $v$ the velocity variable, $p$ the momentum variable, and $a$ the acceleration variable, all elements of $\R^m$. We adopt the notation $q^{(k)}$ to refer to the $k$-th element of $q$, with similar conventions being employed for $v$ and $p$. Denoting $z=(q, p)$; we call $z$ a point in phase space. We denote by $\mathrm{PD}(m)$ the set of $m\times m$ positive definite matrices. Given a map $\Phi : \R^m\to\R^m$, we use the notation $\Phi^k$ to mean the composition $\Phi\circ\cdots\circ\Phi$ ($k$ times). We write $\mathrm{Id}_m$ to denote the identity matrix of size $m\times m$. We denote the Borel $\sigma$-algebra on $\R^m$ by $\mathfrak{B}(\R^m)$.

\subsection{Lagrangian Monte Carlo: The Essentials}\label{lagrangian-remarks:subsec:essentials}

Lagrangian Monte Carlo (LMC) is a geometric method of Bayesian inference that seeks to incorporate second-order information about the posterior in order to produce effective proposals, similar to Riemannian manifold Hamiltonian Monte Carlo (RMHMC). In  this section we review the fundamentals of these methods.
For applications in HMC, an important class of Hamiltonians have the following form:
\begin{definition}\label{lagrangian-remarks:def:riemannian-hamiltonian}
The {\it Riemannian Hamiltonian}
\begin{align}
  \label{lagrangian-remarks:eq:hamiltonian-form} H(q, p) = U(q) + K(q, p)
\end{align}
where $U : \R^m\to\R$ is called the {\it potential energy} function and $K:\R^m\times\R^m\to\R$ is the {\it kinetic energy} function, having the form $K(q, p) = \frac{1}{2} p^\top \mathbf{G}^{-1}(q) p$, where $\mathbf{G} :\R^m\to \mathrm{PD}(m)$ where $\mathbf{G}$ is called the {\it metric}.
\end{definition}
The Riemannian metric defines Christoffel symbols which convey information about the curvature and shape of $\R^m$ imbued with the metric $\mathbf{G}$.
\begin{definition}\label{lagrangian-remarks:def:christoffel-symbols}
  The Christoffel symbols are the $m^3$ functions defined by
  \begin{align}
    \label{lagrangian-remarks:eq:christoffel} \Gamma^k_{ij}(q) = \frac{1}{2} \sum_{l=1}^m \mathbf{G}^{-1}_{kl}(q) \paren{\frac{\partial}{\partial q^{(i)}} \mathbf{G}_{lj}(q) + \frac{\partial}{\partial q^{(j)}} \mathbf{G}_{li}(q) - \frac{\partial}{\partial q^{(l)}} \mathbf{G}_{ij}(q)}.
  \end{align}
\end{definition}
We note that the Christoffel symbols are symmetric in their lower indices (i.e. $\Gamma^k_{ij}(q) = \Gamma^k_{ji}(q)$). The Christoffel symbols play a prominent role in the development of Lagrangian mechanics. As a notational convenience, we will define the matrix-valued function $\Omega : \R\times\R^m\times\R^m\to \R^{m\times m}$ whose $(i, j)$-th entry is
\begin{align}
    \label{lagrangian-remarks:eq:omega}\Omega_{ij}(\epsilon, q, v) = \frac{\epsilon}{2} \sum_{k=1}^m \Gamma^i_{kj}(q) v^{(k)}.
\end{align}
The Riemannian Hamiltonian in \cref{lagrangian-remarks:def:riemannian-hamiltonian} produces equations of motion (see \cref{lagrangian-remarks:subsec:hamiltonian-lagrangian-mechanics} for details), which do not have closed-form solutions. This necessitates the use of numerical integrators.
We now review two integrators that form the basis of our evaluations: the generalized leapfrog integrator, which is a reversible, volume-preserving, and second-order accurate, and the Lagrangian leapfrog method of \citet{lan2015}.
\begin{definition}\label{lagrangian-remarks:def:generalized-leapfrog}
  The {\it generalized leapfrog integrator} for the Hamiltonian equations of motion in \cref{lagrangian-remarks:eq:hamiltonian-position,lagrangian-remarks:eq:hamiltonian-momentum} is a map $(q, p)\mapsto (\tilde{q}, \tilde{p})$ defined by,
  \begin{align}
      \label{lagrangian-remarks:eq:generalized-leapfrog-momentum-i} \breve{p}^{(k)} &= p^{(k)} - \frac{\epsilon}{2} \paren{-\frac{1}{2} \breve{p}^\top \mathbf{G}^{-1}(q) \paren{\frac{\partial \mathbf{G}}{\partial q^{(k)}}(q)} \mathbf{G}^{-1}(q) \breve{p} + \frac{\partial U}{\partial q^{(k)}}(q)} \\
      \label{lagrangian-remarks:eq:generalized-leapfrog-position} \tilde{q} &= q + \frac{\epsilon}{2} \paren{\mathbf{G}^{-1}(q)\breve{p} + \mathbf{G}^{-1}(\tilde{q})\breve{p}} \\
      \label{lagrangian-remarks:eq:generalized-leapfrog-momentum-ii} \tilde{p}^{(k)} &= \breve{p}^{(k)} - \frac{\epsilon}{2} \paren{-\frac{1}{2} \breve{p}^\top \mathbf{G}^{-1}(q) \paren{\frac{\partial \mathbf{G}}{\partial q^{(k)}}(q)} \mathbf{G}^{-1}(q) \breve{p} + \frac{\partial U}{\partial q^{(k)}}(\tilde{q})}.
  \end{align}
\end{definition}
Pseudo-code implementing the generalized leapfrog algorithm is given in \cref{lagrangian-remarks:alg:generalized-leapfrog} in \cref{app:algorithms}.
\begin{definition}\label{lagrangian-remarks:def:lagrangian-leapfrog}
  The {\it Lagrangian leapfrog integrator} for the Lagrangian equations of motion given in \cref{lagrangian-remarks:eq:lagrangian-acceleration} is a map $(q, v)\mapsto (\tilde{q}, \tilde{v})$ defined by, 
  \begin{align}
      \label{lagrangian-remarks:eq:lagrangian-velocity-i} \breve{v} &= \left[\mathrm{Id}_m + \Omega(\epsilon, q, v)\right]^{-1} \left[v - \frac{\epsilon}{2} \mathbf{G}^{-1}(q) \nabla U(q)\right] \\
      \label{lagrangian-remarks:eq:lagrangian-position-update} \tilde{q} &= q + \epsilon ~\breve{v} \\
      \label{lagrangian-remarks:eq:lagrangian-velocity-ii} \tilde{v} &= \left[\mathrm{Id}_m + \Omega(\epsilon, \tilde{q}, \breve{v})\right]^{-1} \left[\breve{v} - \frac{\epsilon}{2} \mathbf{G}^{-1}(\tilde{q}) \nabla U(\tilde{q})\right].
  \end{align}
\end{definition}
Pseudo-code implementing the Lagrangian leapfrog algorithm is given in \cref{lagrangian-remarks:alg:lagrangian-leapfrog} in \cref{app:algorithms}. Unlike the generalized leapfrog integrator (\cref{lagrangian-remarks:def:generalized-leapfrog}), which is a symplectic transformation and therefore necessarily volume-preserving, the Lagrangian leapfrog (\cref{lagrangian-remarks:def:lagrangian-leapfrog}) is not volume-preserving. Its Jacobian determinant is computed in \cref{lagrangian-remarks:eq:lagrangian-jacobian-determinant}. A thorough treatment of numerical integrators is provided in \cref{lagrangian-remarks:subsec:numerical-integrators} In the context of Monte Carlo, this means that Markov chains constructed from repeated applications of the Lagrangian leapfrog integrator will require a Jacobian determinant computation, whereas methods based on the generalized leapfrog will not (its Jacobian determinant is one). Such a Markov chain is the subject of the following example.
\begin{example}
Let $H:\R^m\times\R^m \to\R$ be as in \cref{lagrangian-remarks:def:riemannian-hamiltonian} and define a probability density $\pi(q, p)\propto \exp(-H(q, p))$. Let $q\in\R^m$ be given; a single Markov chain step is constructed as follows. Sample $p \sim \mathrm{Normal}(0, \mathbf{G}(q))$. Fix $k\in\mathbb{N}$. We consider RMHMC and LMC separately:
\begin{description}
\item[LMC] Let $\tilde{\Phi}_\epsilon$ denote the Lagrangian leapfrog (\cref{lagrangian-remarks:def:lagrangian-leapfrog}). Compute the proposal $(\tilde{q}, \tilde{v}) = \tilde{\Phi}(q, \mathbf{G}^{-1}(q) p)$ and set $\tilde{p} = \mathbf{G}(\tilde{q}) \tilde{v}$. Compute the Jacobian determinant $J$ of the map $(q, p)\mapsto (\tilde{q},\tilde{p})$ using \cref{lagrangian-remarks:eq:lagrangian-jacobian-determinant}.
\item[RMHMC] Let $\hat{\Phi}_\epsilon$ denote the generalized leapfrog integrator (\cref{lagrangian-remarks:def:generalized-leapfrog}). Compute the proposal $(\tilde{q}, \tilde{p}) = \hat{\Phi}(q, p)$ and set $J=1$.
\end{description}
Accept the proposal state $\tilde{q}$ with probability $\alpha((q, p), (\tilde{q},\tilde{p}), J) = \min\set{1, \frac{\pi(\tilde{q}, \tilde{p})}{\pi(q, p)} \cdot J}$; otherwise remain at the current state $q$.
\end{example}
A more rigorous treatment of the LMC and RMHMC Markov chains is given in \cref{lagrangian-remarks:subsec:hamiltonian-lagrangian-monte-carlo} using the framework of diffeomorphism Monte Carlo. Pseudo-code is provided in \cref{lagrangian-remarks:alg:hmc} in \cref{app:algorithms}.

\section{Related Work}\label{lagrangian-remarks:sec:related-work}

The focus of the present work is to investigate the numerical methods of integration that were proposed in \citet{lan2015}. \Citet{pmlr-v139-brofos21a} gave an evaluation of the implicit midpoint integrator for RMHMC with special attention paid to the errors in reversibility and volume preservation that were produced by the implicit midpoint algorithm compared to the generalized leapfrog method, as well as the energy conservation properties enjoyed by the implicit midpoint integrator. Other mechanisms of explicit integration have been considered with applications to RMHMC foremost in mind, such as \citet{cobb2019introducing} which produced a reversible, volume-preserving numerical method in an expanded phase-space. Due to the expansion of phase-space, this integrator cannot be used to produce a Markov chain satisfying detailed balance. The work of \citet{NIPS2014_a87ff679} explored alternating blockwise Metropolis-within-Gibbs-like strategies with Riemannian metrics chosen to produce separable Hamiltonians within each block; each block can then be integrated using the standard leapfrog integrator. LMC has previously been criticized in the literature for having unfavorable performance in high dimensions; this failure of LMC in relation to RMHMC is discussed in \citet{geometric-foundations} and we will see evidence of this degradation in \cref{lagrangian-remarks:subsec:student-t}.

\section{Analytical Apparatus}\label{lagrangian-remarks:sec:analytical-apparatus}

In this section we describe an algorithmic recommendation for the numerical integrator used in Lagrangian Monte Carlo and we clarify certain statements around the order of this numerical method. Specifically, we show how to reduce the number of determinant computations from four to two, and that the integrator of Lagrangian dynamics has third-order local error, comparable to the error of the (generalized) leapfrog method used in HMC.

\subsection{Inversion of the Integration Sequence}

A disadvantage of the Lagrangian integrator is that it involves four Jacobian determinant computations at each step of the integrator. In general, computing the Jacobian determinant of an $m\times m$ matrix incurs a computational cost like $\mathcal{O}(m^3)$. Therefore, it seems worthwhile to investigate mechanisms by which to reduce the number of these calculations that are required. \Citet{lan2015} proposed one method that has only two Jacobian determinant computations, but necessitates the return to implicit methods of integration. 
To retain the advantages of explicit integration, we propose a method that computes two Jacobian determinants in every step and consists only of {\it explicit} integration steps.
To achieve this, we propose a conceptually simple procedure: invert the sequence in which position and velocity are updated in the Lagrangian integrator so that position is updated twice at the beginning and end of each step and velocity is updated once in between each update to position. Formally:
\begin{definition}\label{lagrangian-remarks:def:inverted-lagrangian-leapfrog}
  The {\it inverted Lagrangian leapfrog integrator} for the Lagrangian equations of motion given in \cref{lagrangian-remarks:eq:lagrangian-acceleration} is a map $(q, v)\mapsto (\tilde{q}, \tilde{v})$ defined by, 
  \begin{align}
       \label{lagrangian-remarks:eq:inverted-lagrangian-position-i} \breve{q} &= q + \frac{\epsilon}{2} ~v \\
      \label{lagrangian-remarks:eq:inverted-lagrangian-velocity} \tilde{v} &= \left[\mathrm{Id}_m + 2\Omega(\epsilon, \breve{q}, v)\right]^{-1} \left[v - \epsilon \mathbf{G}^{-1}(\breve{q}) \nabla U(\breve{q})\right] \\
      \label{lagrangian-remarks:eq:inverted-lagrangian-position-ii} \tilde{q} &= \breve{q} + \frac{\epsilon}{2} \tilde{v}.
  \end{align}
\end{definition}
Pseudo-code for this procedure is provided in \cref{lagrangian-remarks:alg:inverted-lagrangian-leapfrog} in \cref{app:algorithms}. The basic modification requires only two Jacobian determinant computations per step since the two updates to position in \cref{lagrangian-remarks:eq:inverted-lagrangian-position-i,lagrangian-remarks:eq:inverted-lagrangian-position-ii}, being shear transformations, are volume-preserving in $(q, v)$-space \citep{modi2021delayed}. The required change in volume due to the mapping $(q, p)\mapsto (\tilde{q}, \tilde{p})$ is readily obtained as
\begin{align}
  \begin{split}
    \label{lagrangian-remarks:eq:inverted-lagrangian-jacobian-determinant} &\abs{\mathrm{det}\paren{\frac{\partial (\tilde{q}, \tilde{p})}{\partial (q, p)}}} = \abs{\frac{\mathrm{det}(\mathbf{G}(\tilde{q}))}{\mathrm{det}(\mathbf{G}(q))} \frac{\mathrm{det}(\mathrm{Id}_m - \Omega(\epsilon, \breve{q}, \tilde{v}))}{\mathrm{det}(\mathrm{Id}_m + \Omega(\epsilon, \breve{q}, v))}}.
  \end{split}
\end{align}
Although this Jacobian determinant differs from that produced by \cref{lagrangian-remarks:alg:lagrangian-leapfrog} in \cref{app:algorithms}, we still have the following important property.
\begin{lemma}\label{lagrangian-remarks:lem:inverted-lagrangian-properties}
  The inverted Lagrangian leapfrog integrator is self-adjoint and has at least first-order local error.
\end{lemma}
A proof is given in \cref{lagrangian-remarks:app:proofs-concerning-numerical-order}.
\begin{corollary}
The inverted Lagrangian leapfrog integrator has at least second-order local error.
\end{corollary}
\begin{proof}
This follows as an immediate corollary of \cref{lagrangian-remarks:lem:inverted-lagrangian-properties}
\end{proof}
\begin{definition}\label{lagrangian-remarks:def:ilmc}
  Let $\Phi_\epsilon$ be the inverted Lagrangian leapfrog integrator with step-size $\epsilon\in\R$ (\cref{lagrangian-remarks:def:inverted-lagrangian-leapfrog}). Let $k\in\mathbb{N}$ be the number of integration steps. The {\it inverted Lagrangian Monte Carlo (ILMC)} is an instance involutive Monte Carlo (\cref{lagrangian-remarks:def:involutive-monte-carlo}) with involution $\mathbf{F} \circ \Phi_{\epsilon}^k$ where $\mathbf{F}$ is the momentum flip operator (\cref{lagrangian-remarks:def:momentum-flip}).
\end{definition}
Let us denote by $\hat{\Phi}_\epsilon$ the Lagrangian leapfrog integrator (\cref{lagrangian-remarks:def:lagrangian-leapfrog}) and $\check{\Phi}_{\epsilon}$ the inverted Lagrangian leapfrog (\cref{lagrangian-remarks:def:inverted-lagrangian-leapfrog}). Let $\Phi_\epsilon$ be the exact time $\epsilon$ solution of \cref{lagrangian-remarks:eq:lagrangian-acceleration}. Because both methods are second-order accurate, it follows that
\begin{align}
    \Vert \hat{\Phi}_\epsilon - \check{\Phi}_\epsilon \Vert &\leq \Vert \hat{\Phi}_\epsilon - \Phi_\epsilon \Vert + \Vert \check{\Phi}_\epsilon - \Phi_\epsilon \Vert \\
    &= \mathcal{O}(\epsilon^3)
\end{align}
Intuitively, in the limit of small step-sizes, the difference in proposals generated by the Lagrangian leapfrog and the inverted Lagrangian leapfrog will be minimal, but ILMC (\cref{lagrangian-remarks:def:ilmc}) involves half the number of Jacobian determinant computations compared to LMC (\cref{lagrangian-remarks:def:lmc}) and should therefore be preferred. On the other hand, for large step-sizes, the situation is less clear.

\subsection{Aversions to Inverting the Integration Sequence}

\begin{figure}
  \centering
  \includegraphics[width=0.49\textwidth]{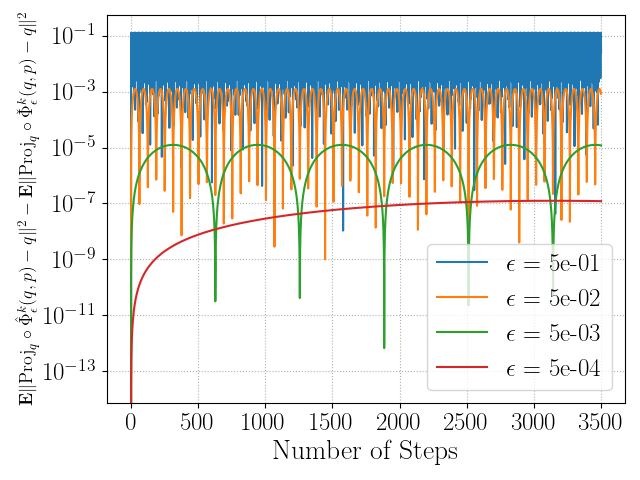}
  \caption{Visualization of the difference in expected distance when integrating a trajectory of a separable, quadratic Hamiltonian for a varying number of steps and integration step-sizes. This difference is always non-negative, indicating that one expects a greater distance between position variables when integrating with the leapfrog, rather than inverted leapfrog, algorithm.}
  \label{lagrangian-remarks:fig:position-difference}
\end{figure}

In the case when $\mathbf{G} = \mathrm{Id}_m$, the Lagrangian integrator devolve into the standard leapfrog integrator (\cref{lagrangian-remarks:def:standard-leapfrog}) that is ubiquitous in HMC. There are good reasons why implementations of HMC integrate in the order of an initial half-step in momentum, a full step in position, and then a second half-step of momentum. This is described visually by \citet{bou-rabee_sanz-serna_2018} for the Hamiltonian $H(q, p) = q^2 / 2 + p^2 / 2$, who make the argument that the leapfrog integrator produces high acceptance probabilities along the $q$-axis, which are desirable. By contrast, inverted leapfrog (\cref{lagrangian-remarks:def:inverted-leapfrog}) produces large acceptance probabilities along the $p$-axis, which are not valuable in HMC. 
Here we wish to expand on this example to consider what happens when HMC Markov chains based on the standard leapfrog and inverted leapfrog are initialized at stationarity. Based on the symmetric roles of $q$ and $p$ in the Hamiltonian, it is tempting to conclude that these Markov chains would exhibit similar performance profiles; surprisingly, this is not the case as shown in the following example.
\begin{example}\label{lagrangian-remarks:ex:inverted-order-inferior}
  Consider a Hamiltonian of the form $H(q, p) = \omega^2 q^2 / 2 + p^2 / 2$. This Hamiltonian corresponds to the distributions $q\sim\mathrm{Normal}(0, 1/\omega^2)$ and $p\sim \mathrm{Normal}(0, 1)$. Let $\hat{\Phi}_\epsilon : \R^m\times\R^m\to \R^m\times\R^m$ and $\check{\Phi}_\epsilon : \R^m\times\R^m\to\R^m\times\R^m$ be the single step leapfrog (\cref{lagrangian-remarks:def:standard-leapfrog}) and inverted leapfrog (\cref{lagrangian-remarks:def:inverted-leapfrog}) methods with step-size $\epsilon$, respectively. Let $\mathrm{Proj}_q(q, p) = q$ be the projection onto the $q$-variables. Then,
  \begin{align}
    \mathbb{E} \left[(\mathrm{Proj}_q\circ\hat{\Phi}_\epsilon(q, p) - q)^2\right] &= \frac{\epsilon^4\omega^2}{4} + \epsilon^2 \\
    \mathbb{E} \left[(\mathrm{Proj}_q\circ\check{\Phi}_\epsilon(q, p) - q)^2\right] &= \frac{\epsilon^4\omega^2}{4} + \epsilon^2 \paren{1 - \frac{\epsilon^2\omega^2}{4}}.
  \end{align}
  Both the leapfrog and inverted leapfrog methods are only numerically stable when $\epsilon^2 \omega^2 < 4$ \citep{leimkuhler_reich_2005}. Hence, $\mathbb{E} \left[(\mathrm{Proj}_q\circ\check{\Phi}_\epsilon(q, p) - q)^2\right] < \mathbb{E} \left[(\mathrm{Proj}_q\circ\hat{\Phi}_\epsilon(q, p) - q)^2\right]$, which we interpret to mean that, in a single step of the integrator, the expected squared distance between initial and terminal position is greater for the leapfrog integrator than for the inverted leapfrog. By deriving the propagator matrices (see \cref{lagrangian-remarks:app:propagator-matrices}) of the leapfrog and inverted leapfrog integrators, we may deduce exact formulas for the $k$-step quantities $\mathbb{E} \left[(\mathrm{Proj}_q\circ\hat{\Phi}^k_\epsilon(q, p) - q)^2\right]$ and $\mathbb{E} \left[(\mathrm{Proj}_q\circ\check{\Phi}^k_\epsilon(q, p) - q)^2\right]$. In \cref{lagrangian-remarks:fig:position-difference} we show the difference of these expected squared distances as a function of the number of steps and for several step-sizes. Notably, this difference is always non-negative, indicating that one expects the leapfrog to produce more distant proposals than the inverted leapfrog in the Gaussian case.
\end{example}
Because autocorrelation is related to the distance to subsequent samples, one expects the inverted leapfrog method to exhibit fewer effective samples. This will be true even if LMC Markov chain is initialized in the stationary distribution. This illustrates an important short-coming of inverting the integration sequence, which must be balanced against computational savings obtained by the reduction in Jacobian determinant computations in the LMC setting.

\subsection{The Order of the Explicit Integrators}

In \citet{lan2015}, the authors showed that the local error rate of the Lagrangian integrator is at least $\mathcal{O}(\epsilon^2)$. We now build on this foundation in order to deduce that the local error rate of the Lagrangian leapfrog is at least $\mathcal{O}(\epsilon^3)$. This means that the order of the LMC integrator matches the local and global error rates of the integrators used in HMC and RMHMC.
Proofs of \cref{lagrangian-remarks:lem:second-order-local-error,lagrangian-remarks:lem:self-adjoint} may be found in \cref{lagrangian-remarks:app:proofs-concerning-numerical-order}.

\begin{lemma}\label{lagrangian-remarks:lem:second-order-local-error}
  The explicit integrator of the Lagrangian dynamics has at least first-order local error.
\end{lemma}
\begin{lemma}\label{lagrangian-remarks:lem:self-adjoint}
  The explicit integrator of Lagrangian dynamics is self-adjoint.
\end{lemma}
\begin{proposition}\label{lagrangian-remarks:prop:third-order-local-error}
  The explicit integrator of the Lagrangian dynamics has at least third-order local error.
\end{proposition}
\begin{proof}
  From \cref{lagrangian-remarks:lem:second-order-local-error} we know that $r\geq 1$. Suppose $r=1$. From \cref{lagrangian-remarks:lem:self-adjoint} we know that the integrator is self-adjoint. From \cref{lagrangian-remarks:thm:self-adjoint-even-order} we know that the order of a self-adjoint method must be even; hence $r$ cannot be odd. But $r=1$ by assumption, a contradiction. Therefore, it must be that $r\geq 2$ so that the explicit integrator has, at least, third-order local error.
\end{proof}
Denote by $\hat{\Phi}_\epsilon$ the second-order integrator of Lagrangian dynamics. If $\hat{\Phi}_\epsilon(q_0, v_0) = (\hat{q}_\epsilon, \hat{v}_\epsilon)$, one wonders if $(\hat{q}_\epsilon, \hat{p}_\epsilon)$ is a second-order approximation of the {\it Hamiltonian dynamics}, where $\hat{p}_\epsilon = \mathbf{G}(\hat{q}_\epsilon) \hat{v}_\epsilon$. Indeed, this is true and follows as an immediate consequence of \cref{lagrangian-remarks:prop:composition-order} with diffeomorphism $(q, v)\mapsto (q, \mathbf{G}v)$.

\subsection{Efficient Computation of the Jacobian Determinant and the Update to Velocity}

The update to the velocity in \cref{lagrangian-remarks:eq:inverted-lagrangian-velocity,lagrangian-remarks:eq:lagrangian-velocity-i,lagrangian-remarks:eq:lagrangian-velocity-ii} and the associated change-in-volume in \cref{lagrangian-remarks:eq:inverted-lagrangian-jacobian-determinant,lagrangian-remarks:eq:inverted-lagrangian-jacobian-determinant} involve manipulations of a matrix of the form $\mathrm{Id} + \Omega(\epsilon, q, v)$: in the former case, we must solve a linear system involving this matrix while in the latter case we must compute the absolute value of its Jacobian determinant. An efficient procedure by which to achieve both of these objectives is to compute the PLU decomposition of $\mathrm{Id} + \Omega(\epsilon, q, v)=\mathbf{P}\mathbf{L}\mathbf{U}$ where $\mathbf{P}$ is a permutation matrix, $\mathbf{L}$ is a lower-triangular matrix with unit diagonal, and $\mathbf{U}$ is an upper-triangular matrix. The computational cost of this decomposition is $\frac{2}{3}m^3 + \mathcal{O}(m^2)$. Linear systems can be solved using the PLU decomposition by applying the permutation and solving the triangular systems via forward-backward substitution. Moreover, the required Jacobian determinant is simply $\prod_{i=1}^m \mathbf{U}_{ii}$, since the determinant of the permutation matrix has unit magnitude and the lower-triangular matrix $\mathbf{L}$ has unit Jacobian determinant since all of its diagonal elements are equal to one.

\subsection{Built-In Robustness of the Lagrangian Integrator}\label{lagrangian-remarks:subsec:robustness}

As shown in \citet{brofos2021numerical}, the volume-preservation property (i.e. $J=1$ in \cref{lagrangian-remarks:alg:generalized-leapfrog} in \cref{app:algorithms}) of the generalized leapfrog integrator is predicated on the symmetry of partial derivatives $\nabla_q^\top \nabla_p H(q, p) = \nabla_p \nabla_q H(q,p)$. In an {\it implementation} of the generalized leapfrog integrator we may suppose that we have functions $g_{k}(q)$ representing $\frac{\partial \mathbf{G}}{\partial q^{(k)}}(q)$. Substituting this function into the definition of the generalized leapfrog integrator (\cref{lagrangian-remarks:def:generalized-leapfrog}) yields the following map $(q, p)\mapsto (\tilde{q}, \tilde{p})$
\begin{align}
    \breve{p}^{(k)} &= p^{(k)} - \frac{\epsilon}{2} \delta_k(q, \breve{p}) \\
    \tilde{q} &= q + \frac{\epsilon}{2} \paren{\Delta(q, \breve{p}) + \Delta(\tilde{q}, \breve{p})} \\
    \tilde{p}^{(k)} &= \breve{p}^{(k)} - \frac{\epsilon}{2} \delta_k(\tilde{q}, \breve{p}),
\end{align}
where
\begin{align}
    \delta_k(q, p) &= -\frac{1}{2} p^\top \mathbf{G}^{-1}(q) g_k(q) \mathbf{G}^{-1}(q) p + \frac{\partial U}{\partial q^{(k)}}(q) \\
    \Delta(q, p) &= \mathbf{G}^{-1}(q)p
\end{align}
When $g_k(q) = \frac{\partial \mathbf{G}}{\partial q^{(k)}}(q)$, the resulting map is necessarily volume-preserving. However, we may then ask the question, ``What happens when $g_k(q)$ is {\it incorrectly implemented} so that, in fact, $g_k(q) \neq \frac{\partial \mathbf{G}}{\partial q^{(k)}}(q)$?'' The symmetry of partial derivatives has therefore been violated since 
\begin{align}
    \frac{\partial\Delta}{\partial q^{(k)}}(q, p) = \frac{\partial \mathbf{G}^{-1}}{\partial q^{(k)}}(q) p \neq -\mathbf{G}^{-1}(q) g_k(q) \mathbf{G}^{-1}(q) p = \frac{\partial \delta_k}{\partial p}(q, p).
\end{align}
In the case of the RMHMC Markov chain (\cref{lagrangian-remarks:def:rmhmc}), detailed balance is no longer satisfied and there is no expectation that the RMHMC will converge to the target distribution.

The situation is different in the case of the LMC Markov chain (\cref{lagrangian-remarks:def:lmc}). The fundamental difference is that LMC {\it expects} the transformation $(q, v)\mapsto (\tilde{q},\tilde{v})$ to be non-volume-preserving, hence necessitating the Jacobian determinant correction in \cref{lagrangian-remarks:eq:acceptance-probability}. To see that the change-in-volume is still correctly computed even when $g_k(q) \neq \frac{\partial \mathbf{G}}{\partial q^{(k)}}(q)$, we observe that the Lagrangian leapfrog's (\cref{lagrangian-remarks:def:lagrangian-leapfrog}) update to velocity in \cref{lagrangian-remarks:eq:lagrangian-velocity-i} is a special case of the following map:
\begin{align}
    \breve{v} = \paren{\mathrm{Id}_m + \mathbf{A}(q, v)}^{-1} \paren{v - \mathbf{b}(q)}
\end{align}
where $\mathbf{A} : \R^m\times\R^m\to\R^{m\times m}$ and $\mathbf{b} : \R^m\to\R^m$. Under the assumption that $\mathbf{A}(q, v)\tilde{v} = \mathbf{A}(q, \tilde{v}) v$ (which holds for the correctly implemented LMC with $\mathbf{A} = \Omega(\epsilon, q, v)$ by \cref{lagrangian-remarks:prop:omega-properties}), the Jacobian determinant of the map $(q, v)\mapsto (q, \breve{v})$ is
\begin{align}
    \label{lagrangian-remarks:eq:general-determinant} \abs{\mathrm{det}\paren{\frac{\partial (q, \breve{v})}{\partial (q, v)}}} = \abs{\frac{\mathrm{det}(\mathrm{Id}_m + \mathbf{A}(q, \breve{v}))}{\mathrm{det}(\mathrm{Id}_m + \mathbf{A}(q, v))}}.
\end{align}
Hence, if we adopt the notation $g_{k,ij}(q)$ as the $(i,j)$-th element of $g_k(q)$, and substitute
\begin{align}
    \tilde{\Gamma}^k_{ij}(q) = \frac{1}{2} \sum_{l=1}^m \mathbf{G}^{-1}_{kl}(q) \paren{g_{i,lj}(q) + g_{j,li}(q) - g_{l, ij}(q)}
\end{align}
for \cref{lagrangian-remarks:eq:christoffel} and define $\tilde{\Omega}_{ij}(\epsilon, q, v) = \frac{\epsilon}{2} \sum_{k=1}^m \tilde{\Gamma}^i_{kl}(q) v^{(k)}$ then we still have $\tilde{\Omega}(\epsilon, q, v) \breve{v} = \tilde{\Omega}(\epsilon, q, \breve{v})v$ using the fact that $\tilde{\Gamma}^k_{ij}(q) = \tilde{\Gamma}^k_{ji}(q)$. Hence, with $\mathbf{A}(q, v) = \tilde{\Omega}(\epsilon, q, v)$, \cref{lagrangian-remarks:eq:general-determinant} applies to compute the Jacobian determinant.

\section{Experimentation}\label{lagrangian-remarks:sec:experimentation}

We now turn our attention to the empirical evaluation of the numerical integrator of Lagrangian dynamics in terms of its numerical order, its inverted variant, and its robustness to misspecification of the derivatives of the metric. We begin in \cref{lagrangian-remarks:subsec:second-order} by numerically validating the second-order behavior of the numerical integrator. In the subsequent material, we evaluate the integrator with and without inversion in a banana-shaped distribution, in Bayesian logistic regression, in a multiscale Student-$t$ distribution, and in a stochastic volatility model. As baselines, we consider RMHMC and HMC. Code to reproduce these experiments may be found at
\if0\blind {
\url{https://github.com/JamesBrofos/Rethinking-Lagrangian-Monte-Carlo}.
} \fi
\if1\blind {
\url{https://github.com/Anonymous/GitHubLink}.
} \fi

We consider three metrics by which to assess the convergence of the Markov chain produced by ILMC and the baselines. First, we consider the expected squared jump distance (ESJD) as described in \citet{esjd}; this measures the expected squared distance between the current state and the next state, where the expectation is computed over the acceptance probability. The larger the ESJD, the less serial autocorrelation in the Markov chain samples. We also consider the effective sample size (ESS) normalized by time elapsed, which gives an indication of the sampling efficiency of each method. We use the implementation of ESS as given in \citet{arviz_2019}. We also consider the method of \citet{brofos2021numerical} for measuring the ergodicity of the Markov chain given i.i.d. samples. Under this procedure, we project the i.i.d. samples and the Markov chain samples along one-hundred random directions and measure the average value of the Kolmogorov-Smirnov statistics of these one-dimensional distributions. By the Cram\'{e}r-Wold theorem, the closer these Kolmogorov-Smirnov statistics are concentrated toward zero, the higher the fidelity between the Markov chain samples and the i.i.d. samples. In implementing the generalized leapfrog integrator (\cref{lagrangian-remarks:def:generalized-leapfrog}), we resolve the fixed point equations \cref{lagrangian-remarks:eq:generalized-leapfrog-momentum-i,lagrangian-remarks:eq:generalized-leapfrog-position} using fixed point iteration to a convergence tolerance of $1\times 10^{-6}$, with convergence measured in $\Vert\cdot\Vert_\infty$.

\subsection{Demonstration of Second-Order Error}\label{lagrangian-remarks:subsec:second-order}

\begin{figure}
  \centering
  \includegraphics[width=0.49\textwidth]{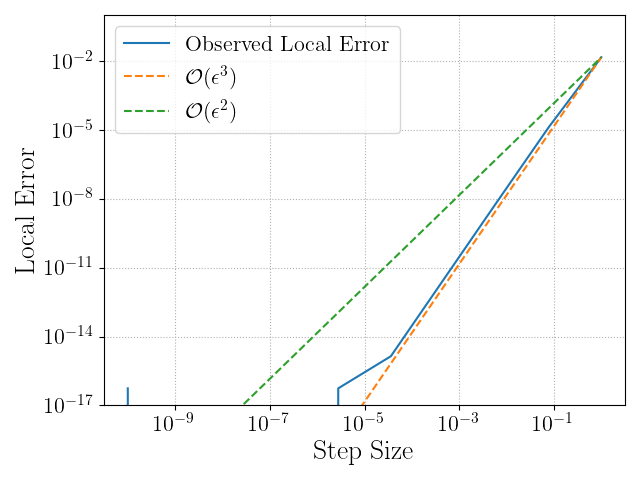}
  \caption{Visualization of the second-order accuracy (i.e. third order local error) of the Lagrangian integrator on Lagrangian dynamics derived from a non-separable Hamiltonian. We see that the observed local error decreases on a log-log scale with a slope of three, corresponding to the claimed third-order local error. By contrast, we also show a line with a slope of two, indicating that the integrator exhibits accuracy better than first-order.}
  \label{lagrangian-remarks:fig:second-order-error}
\end{figure}

We consider the following non-separable Hamiltonian $H(q, p) = \frac{q^2 p^2}{2}$, which describes geodesic motion in $\R$ when equipped with the metric $G(q) = 1 / q^2$. After converting from momentum to velocity $v_t = \frac{p_t}{q_t^2}$, we obtain the second-order differential equation $a_t = \frac{v_t^2}{q_t}$. Given initial conditions $q_0$ and $p_0$ (in the Hamiltonian formalism), the exact solution to this differential equation is $q_t = q_0 \exp(q_0 p_0 t)$ and $v_t = q_0^2p_0 \exp(q_0 p_0 t)$. If, as claimed, the Lagrangian integrator is indeed second-order, then it should exhibit third-order local error according to \cref{lagrangian-remarks:def:integrator-order}. We evaluate this by examing the squared error between the output of the numerical integrator and the analytical solution for a decreasing sequence of step-sizes $\epsilon$; that is, denoting the output of a single step of the Lagrangian integrator by $(\hat{q}_\epsilon,\hat{v}_\epsilon)$, we measure $\Vert \hat{q}_\epsilon - q_\epsilon\Vert_2^2 + \Vert \hat{v}_\epsilon - v_\epsilon\Vert_2^2$. We observe in \cref{lagrangian-remarks:fig:second-order-error} that this error decreases linearly on a log-log scale and, critically, the slope of this linear relation is three. This demonstrates numerically the third-order local error of the Lagrangian integrator and gives support to the claim that the method is of second-order accuracy.

\subsection{Banana-Shaped Posterior Distribution}

\begin{figure*}[t!]
  \centering
  \begin{subfigure}[t]{0.3\textwidth}
  \includegraphics[width=\textwidth]{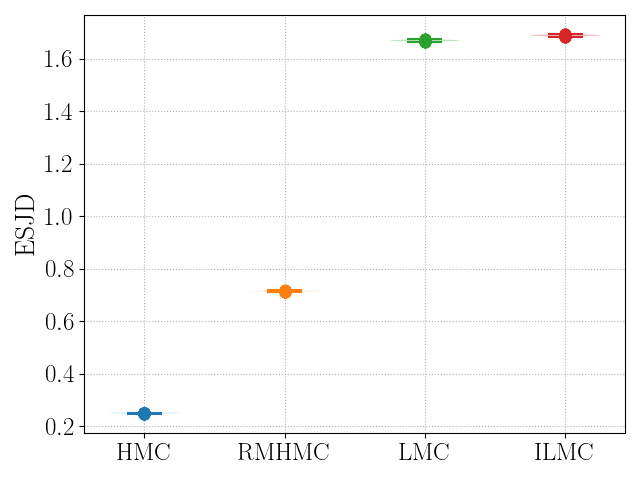}
  \caption{ESJD}
  \end{subfigure}
  ~
  \begin{subfigure}[t]{0.3\textwidth}
  \includegraphics[width=\textwidth]{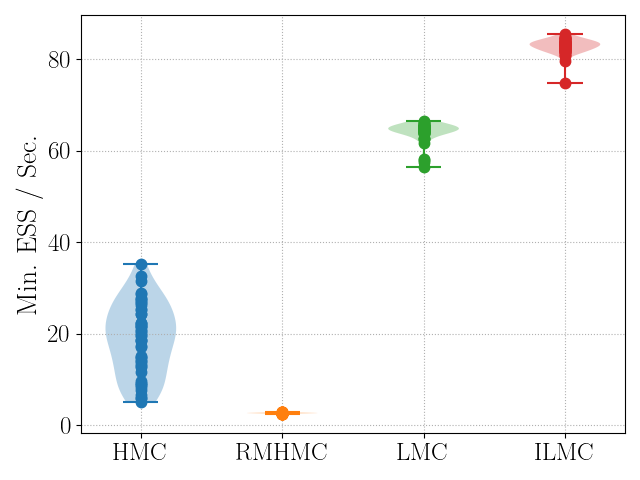}
  \caption{Min. ESS / Sec.}
  \end{subfigure}
  ~
  \begin{subfigure}[t]{0.3\textwidth}
  \includegraphics[width=\textwidth]{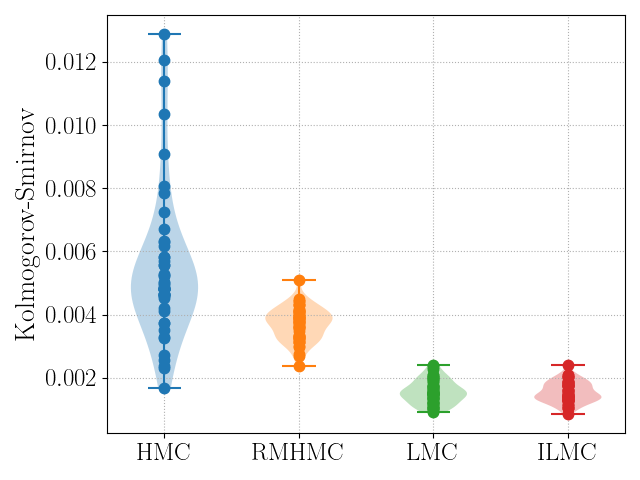}
  \caption{Kolmogorov-Smirnov}
  \end{subfigure}
  \caption{{\bf (Banana-Shaped Distribution)} We show the ESJD, the minimum ESS per second, and the distribution of Kolmogorov-Smirnov statistics for sampling from the banana-shaped distribution using RMHMC, LMC, and ILMC. Surprisingly, we do not observe degradation of the ESJD when employing ILMC. Therefore, combined with its faster sampling iteration, ILMC enjoys the best ESS per second and a distribution of Kolmogorov-Smirnov statistics that is comparable to LMC.}
  \label{lagrangian-remarks:fig:banana}
\end{figure*}

The banana-shaped distribution was proposed in \citet{rmhmc} by Cornebise and Julien as an example of Bayesian inference in non-identifiable models. In this example, a non-identifiable likelihood function in Bayesian linear regression is regularized according to a normal prior, the effect of which is to produce a density with symmetric, elongated tails. The generative model of this distribution is as follows:
\begin{align}
  (\theta_1,\theta_2) &\overset{\mathrm{i.i.d.}}{\sim} \mathrm{Normal}(0, \sigma^2_\theta) \\
  y_i \vert \theta_1,\theta_2 &\overset{\mathrm{i.i.d.}}{\sim} \mathrm{Normal}(\theta_1 + \theta_2^2, \sigma^2_y) ~~~~~ \mathrm{for}~~ i = 1,\ldots, n.
\end{align}
This distribution also illustrates a short-coming of the generalized leapfrog method. For large step-sizes, the implicit update to the momentum variable will not have a solution; therefore, the generalized leapfrog integrator is compelled to adopt a significantly smaller step-size than can be used even by the standard leapfrog method. Indeed, an advantage of explicit numerical integrators is that one does not need to fret that constituent update steps in the integrator will not have solutions. We seek to draw 1,000,000 samples from this posterior.

In our experiments we set $\sigma^2_\theta = \sigma^2_y = 2$, $n = 100$, and set parameter values $\theta_1 = 1/2$ and $\theta_2 = \sqrt{1 - 1/2}$ for generating synthetic data. For HMC we use a step-size of $0.1$ and ten integration steps, which was found to produce an acceptance probability between eighty and ninety percent. As for the Riemannian metric, we adopt the sum of the Fisher information of the log-likelihood and the negative Hessian of the log-prior. For RMHMC, we use a step-size of $0.04$ and twenty integration steps, which produces an acceptance probability of around ninety percent. For LMC and ILMC, we use twenty integration steps with a step-size of $0.1$, which yields an acceptance probability of around ninety percent. These parameter configurations were found to produce reasonable Markov chains based on hand-tuning. Results showing the effective sample size (ESS) per second are provided in \cref{lagrangian-remarks:fig:banana}; we see that RMHMC struggles in this distribution, being even worse than ordinary HMC due to the requirement to use a small step-size. By contrast, LMC and ILMC do significantly better, with ILMC having the best ESS per second due to its elimination of two Jacobian determinant computations. We also show the distribution of this average value over ten trials for each sampling method. We find that the geometric methods based on the Lagrangian formalism perform similarly under this ergodicity measure and outperform competing methods.

\subsection{Bayesian Logistic Regression}

\begin{figure*}[t!]
  \centering
  \begin{subfigure}[t]{0.3\textwidth}
  \includegraphics[width=\textwidth]{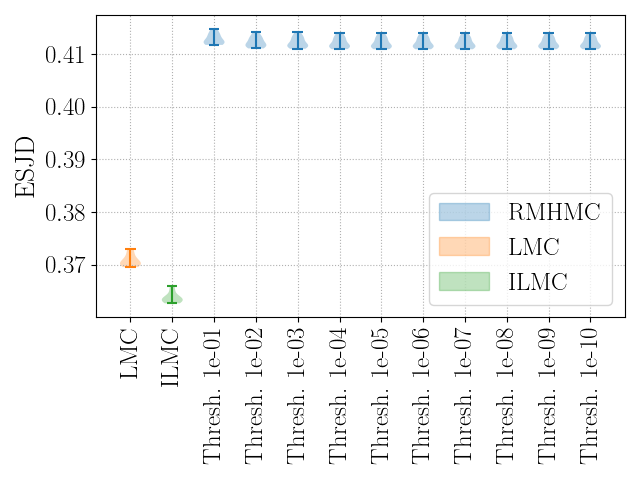}
  \caption{ESJD}
  \end{subfigure}
  ~
  \begin{subfigure}[t]{0.3\textwidth}
  \includegraphics[width=\textwidth]{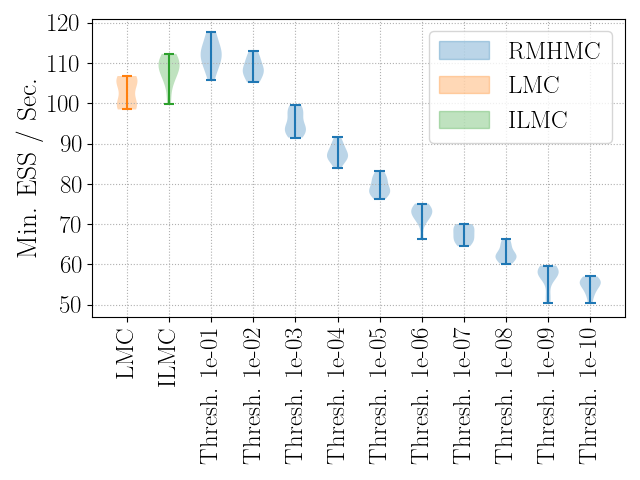}
  \caption{Min. ESS / Sec.}
  \end{subfigure}
  ~
  \begin{subfigure}[t]{0.3\textwidth}
  \includegraphics[width=\textwidth]{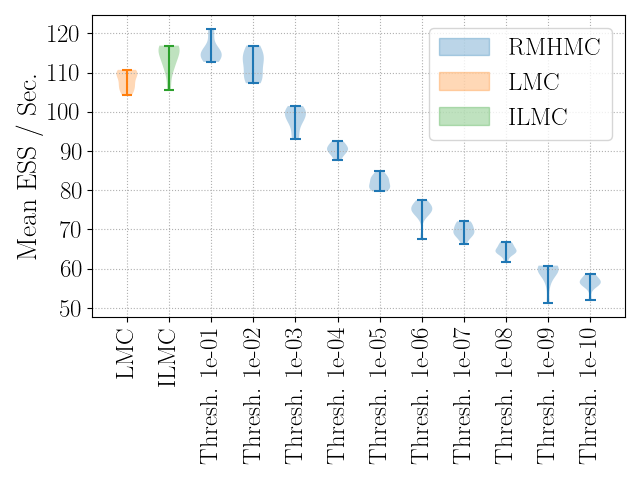}
  \caption{Mean ESS / Sec.}
  \end{subfigure}
  \caption{{\bf (Logistic Regression - Breast Cancer)} We show the ESJD and the minimum and mean ESS per second for the breast cancer dataset using RMHMC, LMC, and ILMC. We observe that LMC and ILMC have degraded movement through the sample space as computed by the ESJD; however, this is compensated for by their superior computational efficiency, ultimately yielding more effectively independent samples per second and RMHMC.}
  \label{lagrangian-remarks:fig:logistic-breast}
\end{figure*}

\begin{figure*}[t!]
  \centering
  \begin{subfigure}[t]{0.3\textwidth}
  \includegraphics[width=\textwidth]{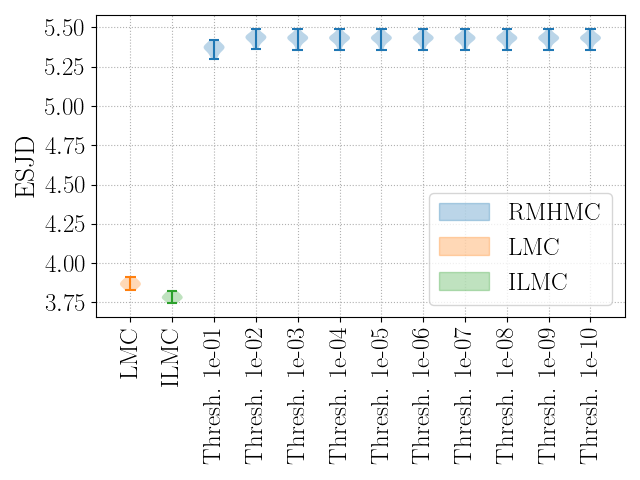}
  \caption{ESJD}
  \end{subfigure}
  ~
  \begin{subfigure}[t]{0.3\textwidth}
  \includegraphics[width=\textwidth]{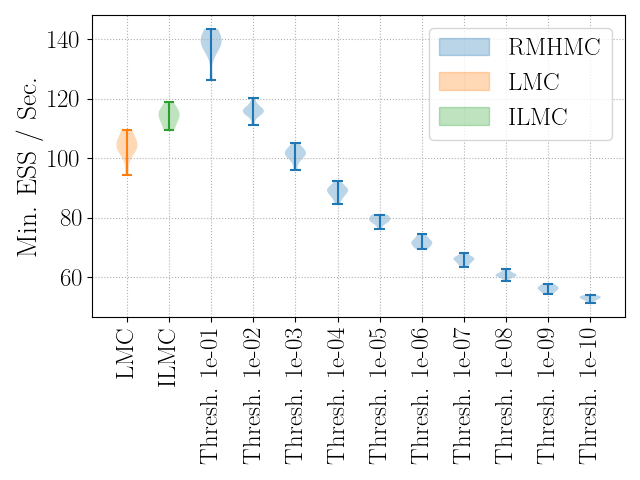}
  \caption{Min. ESS / Sec.}
  \end{subfigure}
  ~
  \begin{subfigure}[t]{0.3\textwidth}
  \includegraphics[width=\textwidth]{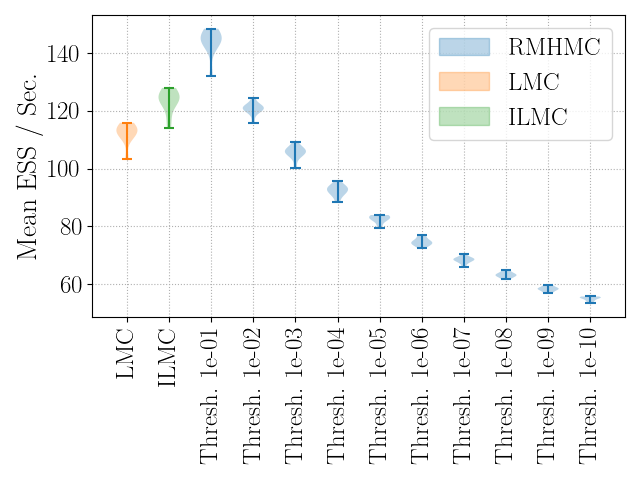}
  \caption{Mean ESS / Sec.}
  \end{subfigure}
  \caption{{\bf (Logistic Regression - Thyroid Cancer)} We show the ESJD and the minimum and mean ESS per second for the thyroid cancer dataset using RMHMC, LMC, and ILMC. We observe that LMC and ILMC have significantly degraded movement through the sample space as computed by the ESJD; however, this is compensated for by their superior computational efficiency, ultimately yielding more effectively independent samples per second and RMHMC.}
  \label{lagrangian-remarks:fig:logistic-thyroid}
\end{figure*}

We consider Bayesian logistic regression on a breast cancer and a thyroid cancer dataset. The breast cancer dataset has 277 observations and ten covariates while the thyroid cancer dataset has 215 observations and six covariates. The Bayesian generative model is assumed to have the following form:
\begin{align}
  \beta &\sim\mathrm{Normal}(0,\alpha^{-1} \mathrm{Id}) \\
  y_i \vert \mathbf{x}_i, \beta &\overset{\mathrm{i.i.d.}}{\sim} \mathrm{Bernoulli}(\sigma(\mathbf{x}_i^\top\beta))~~~~~ \mathrm{for}~~ i = 1,\ldots, n,
\end{align}
where $\sigma : \R\to(0,1)$ denotes the sigmoid function. As the Riemannian metric, we adopt the sum of the Fisher information of the log-likelihood and the negative Hessian of the log-prior. We show in \cref{lagrangian-remarks:fig:logistic-breast,lagrangian-remarks:fig:logistic-thyroid} the minimum ESS per second, where we have also considered varying the convergence threshold used to solve fixed point iterations in RMHMC; as discussed in detail by \citet{brofos2021numerical}, implementations of RMHMC require that the implicit updates to momentum and position be resolved using an iterative procedure such as fixed point iteration or Newton's method. The tolerance in these numerical methods directly controls the degree to which reversibility and volume preservation are violated by the implementation of the generalized leapfrog method: For small values of the threshold, these theoretical properties are closely respected by the numerical method, while for large thresholds, violations occur. Reversibility and volume preservation imply detailed balance in HMC, and violations imply that the detailed balance may not hold in an implementation of RMHMC with large thresholds. By contrast, LMC and ILMC are fully explicit and detailed balance is respected to machine precision for either method. We observe that both sampling methods based on Lagrangian mechanics exhibit higher ESS per second than their counterpart based on the Hamiltonian formalism except for the largest thresholds (with the greatest bias); we refer the interested reader to \citet{brofos2021numerical} for a detailed discussion on the effects of the convergence threshold on the bias of the RMHMC Markov chain. Moreover, ILMC outperforms LMC on this metric. When comparing the ESJD, we observe that, consistent with our understanding in the Gaussian case, the ILMC method has the smallest expected distance traveled. However, this is offset by a faster sampling iteration due to only requiring two Jacobian determinant calculations instead of four.

\subsection{Multiscale Student's $t$-Distribution}\label{lagrangian-remarks:subsec:student-t}

\begin{figure*}[t!]
  \centering
  \begin{subfigure}[t]{0.3\textwidth}
  \includegraphics[width=\textwidth]{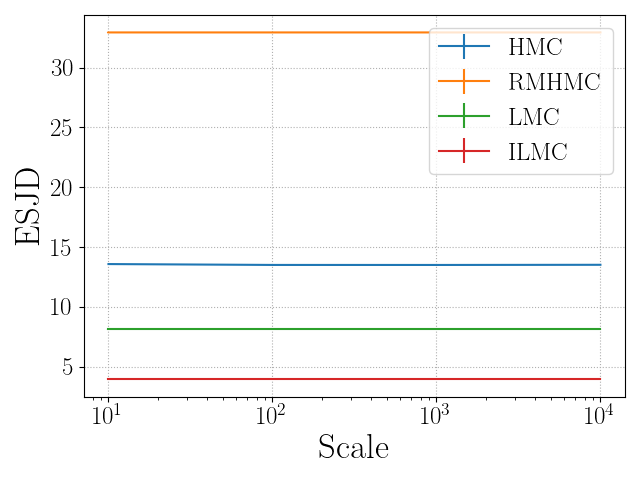}
  \caption{ESJD}
  \end{subfigure}
  ~
  \begin{subfigure}[t]{0.3\textwidth}
  \includegraphics[width=\textwidth]{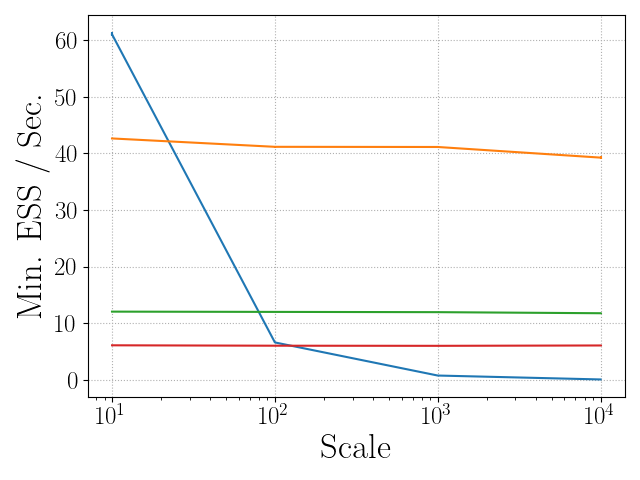}
  \caption{Min. ESS / Sec.}
  \end{subfigure}
  ~
  \begin{subfigure}[t]{0.3\textwidth}
  \includegraphics[width=\textwidth]{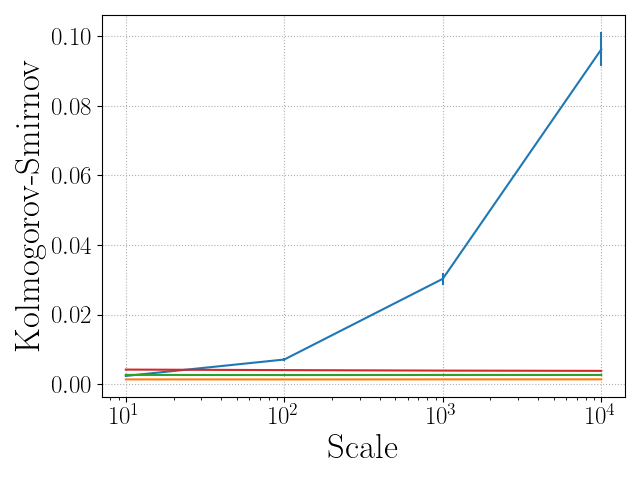}
  \caption{Kolmogorov-Smirnov}
  \end{subfigure}
  \caption{{\bf (Multi-Scale Student-$t$ $\nu = 5\times 10^0$)} We show the ESJD, the minimum ESS per second, and the distribution of Kolmogorov-Smirnov statistics for sampling from the multiscale Student-$t$ distribution with $\nu=5\times 10^0$ using RMHMC, LMC, and ILMC. We observe that RMHMC enjoys the best ESJD, time-normalized ESS, and Kolmogorov-Smirnov statistics. This illustates an important limitation of the Lagrangian procedure in that its performance may actually be unable to match RMHMC in certain circumstances.}
  \label{lagrangian-remarks:fig:t-5}
\end{figure*}

\begin{figure*}[t!]
  \centering
  \begin{subfigure}[t]{0.3\textwidth}
  \includegraphics[width=\textwidth]{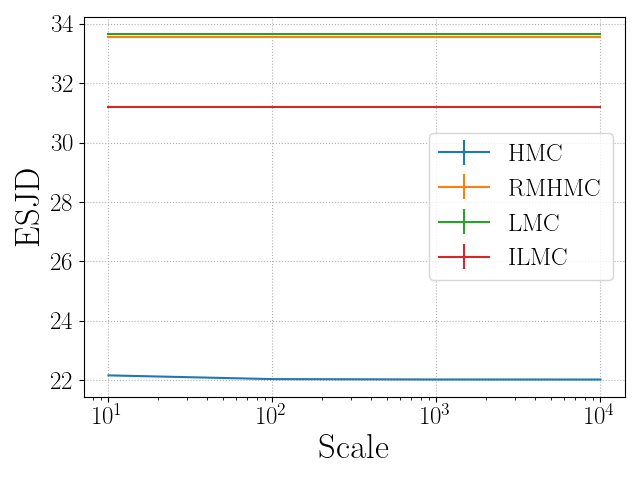}
  \caption{ESJD}
  \end{subfigure}
  ~
  \begin{subfigure}[t]{0.3\textwidth}
  \includegraphics[width=\textwidth]{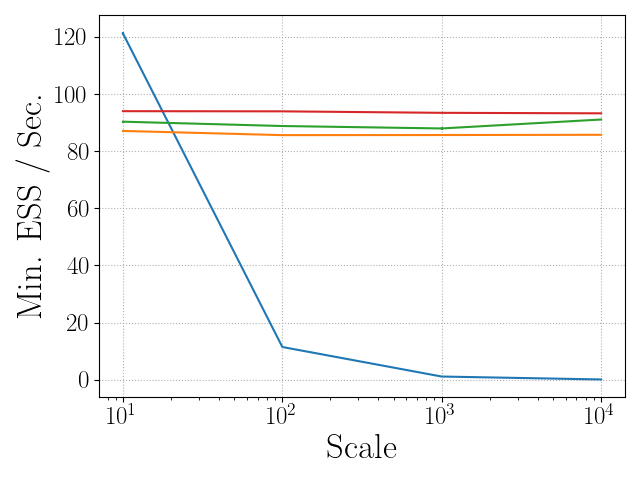}
  \caption{Min. ESS / Sec.}
  \end{subfigure}
  ~
  \begin{subfigure}[t]{0.3\textwidth}
  \includegraphics[width=\textwidth]{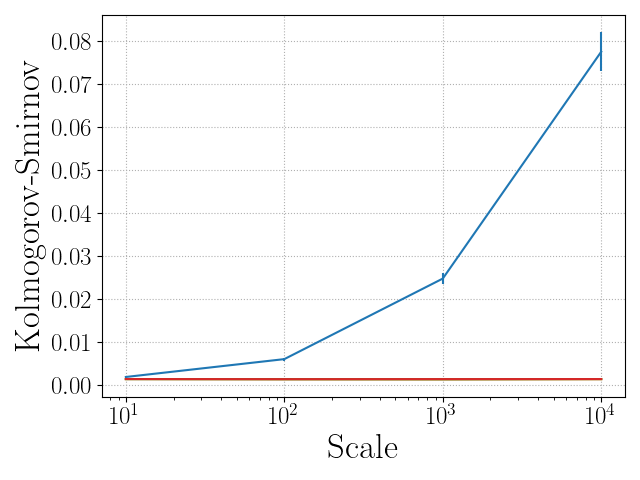}
  \caption{Kolmogorov-Smirnov}
  \end{subfigure}
  \caption{{\bf (Multi-Scale Student-$t$ $\nu = 5\times 10^3$)} We show the ESJD, the minimum ESS per second, and the distribution of Kolmogorov-Smirnov statistics for sampling from the multiscale Student-$t$ distribution with $\nu=5\times 10^3$ using RMHMC, LMC, and ILMC. In this example, despite having a smaller ESJD, ILMC ultimately produces a slight benefit in terms of ESS per second. Among the geometric methods, the distribution of the Kolmogorov-Smirnov statistics are nearly identical.}
  \label{lagrangian-remarks:fig:t-5000}
\end{figure*}

Multiscale distributions can be challenging for HMC since trajectories integrated by the leapfrog method will exhibit significant oscillations along directions of the distributions possessing the smallest spatial scale. To investigate this phenomenon, we consider sampling from a multivariate Student-$t$ distribution with a multiscale covariance; in particular we consider a distribution with density function,
\begin{align}
  \pi(x) \propto \left[1 + \frac{1}{\eta} x^\top\Sigma^{-1}x\right]^{-(\eta + m)/2},
\end{align}
where $x\in\R^m$, $\eta > 2$ is the degrees-of-freedom, and $\Sigma$ is the scale matrix. We consider scale matrices of the form $\Sigma = \mathrm{diag}(1, \ldots, 1, \sigma^2) \in \R^{m\times m}$. In our experiments we set $m=20$ and $\eta \in \set{5\times 10^0, 5\times 10^3}$ and consider multiscale distributions for $\sigma^2\in \set{1\times 10^1, 1\times 10^2,1\times 10^3, 1\times 10^4}$. We choose these two values of the degrees-of-freedom to demonstrate two distinct sampling behaviors. For the Riemannian methods, we consider a step-size of $0.7$ and twenty integration steps. As the Riemannian metric, we use the positive definite term in the negative Hessian of the log-density of the distribution. For $\nu = 5\times 10^0$, we observe that the ESJD is largest for RMHMC; this occurs because RMHMC enjoys a far superior acceptance probability in this scenario at 95\% whereas LMC and ILMC have acceptance probabilities of 40\% and 62\%, respectively. This degraded performance of the Lagrangian methods is then reflected in the time-normalized ESS, which shows RMHMC dominating LMC, ILMC, and HMC. For the case of $\eta = 5\times 10^3$, circumstances are more favorable to the Lagrangian methods, with ILMC exhibiting the best performance in terms of time-normalized ESS, with the LMC method giving the second best results. In terms of ESJD, we observe that ILMC moves less far in sample space than LMC, but that this is offset by the faster sampling. As in the case of the banana-shaped distribution, we may sample from this target density analytically in order to assess the ergodicity properties of the samplers. We find that the ergodicity of the geometric methods is essentially constant with respect to the multiscale parameter, whereas the performance of Euclidean HMC noticeably degenerates. For both $\nu = 5\times 10^3$ and $\nu = 5\times 10^0$, one observes that the performance of the geometric methods is essentially constant over the multiple scales of the target distribution, demonstrating the beneficial effect of capturing the geometry of the target.

\subsection{Stochastic Volatility Model}

\begin{figure*}[t!]
  \centering
  \begin{subfigure}[t]{0.3\textwidth}
  \includegraphics[width=\textwidth]{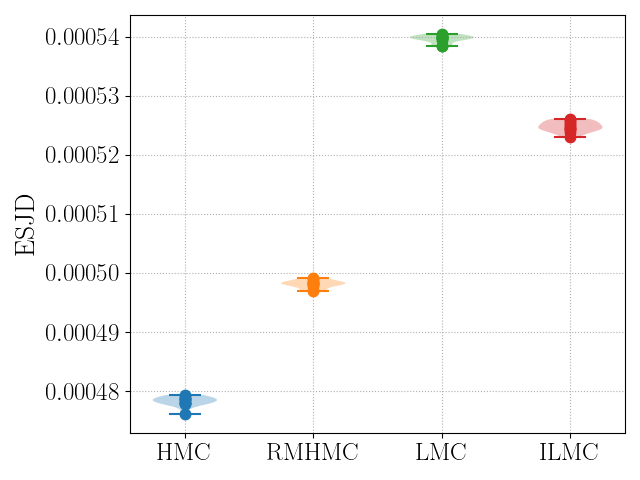}
  \caption{ESJD}
  \end{subfigure}
  ~
  \begin{subfigure}[t]{0.3\textwidth}
  \includegraphics[width=\textwidth]{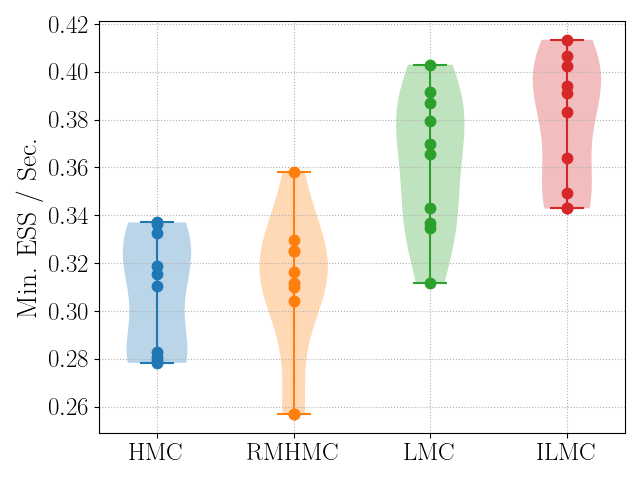}
  \caption{Min. ESS / Sec.}
  \end{subfigure}
  ~
  \begin{subfigure}[t]{0.3\textwidth}
  \includegraphics[width=\textwidth]{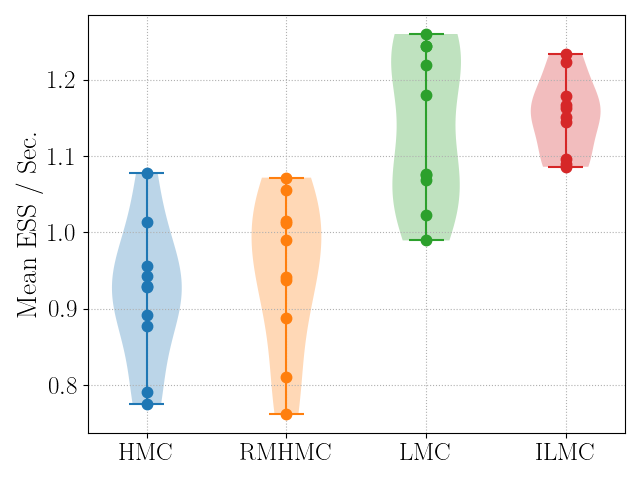}
  \caption{Mean ESS / Sec.}
  \end{subfigure}
  \caption{{\bf (Stochastic Volatility Model)} We show the ESJD and the minimum and mean ESS per second for the stochastic volatility model using HMC, RMHMC, LMC, and ILMC. In this example, ILMC moves less efficiently through the same space as measured by ESJD, and ultimately produces fewer effective samples per second relative to LMC. However, it is intriguing to note that ILMC nonetheless produces more effective transitions than RMHMC and HMC.}
  \label{lagrangian-remarks:fig:stochastic-volatility}
\end{figure*}

We consider Bayesian inference in a stochastic volatility model. We consider the following generative model:
\begin{align}
  \label{lagrangian-remarks:eq:stochastic-volatility-x} x_t\vert x_{t-1}, \phi,\sigma^2 &\sim\mathrm{Normal}(\phi x_{t-1}, \sigma^2) \\
  \label{lagrangian-remarks:eq:stochastic-volatility-y} y_t \vert \beta, x_t &\sim \mathrm{Normal}(0, \beta^2\exp(x_t)) 
\end{align}
for $t = 2,\ldots, T$ in \cref{lagrangian-remarks:eq:stochastic-volatility-x} and $t=1,\ldots, T$ in \cref{lagrangian-remarks:eq:stochastic-volatility-y} and with priors $x_1\sim \mathrm{Normal}(0, \sigma^2 / (1- \phi^2))$, $(\phi+1)/2\sim\mathrm{Beta}(20, 3/2)$, $1/\sigma^2\sim \xi^2(10, 1/20)$, and the prior over $\beta$ being proportional to $1/\beta^2$. Given $(y_1,\ldots, y_T)$, we seek to sample the posterior of $(x_1,\ldots, x_T, \phi, \beta, \sigma^2)$. We follow \citet{rmhmc} and employ a Metropolis-within-Gibbs-like alternating procedure for sampling the posteriors of $(x_1,\ldots, x_T)$ and $(\phi,\beta,\sigma^2)$. In our experiments we set $T=1,000$ and use values of $\phi=0.98$, $\beta=0.65$, and $\sigma^2 = 0.15^2$. For HMC, we use a step-size of 0.01 and fifty integration steps when sampling $(\phi,\beta,\sigma^2)$; for the geometric methods, we use a step-size of 0.5 and six integration steps.
As for the Riemannian metric, we adopt the sum of the Fisher information of the log-likelihood and the negative Hessian of the log-prior. We compare the average ESS per second among the three latent variables $(\phi, \beta, \sigma^2)$ with results reported in \cref{lagrangian-remarks:fig:stochastic-volatility}. We find that LMC and ILMC are the strongest performing methods, with LMC having better ESS per second due to its more efficient traversal of the sample space. Indeed, ILMC has degraded performance in this example, owing to its greater autocorrelation, but nevertheless outperforms HMC and RMHMC.

\subsection{Fitzhugh-Nagumo Model}

\begin{figure*}[t!]
  \centering
  \begin{subfigure}[t]{0.3\textwidth}
  \includegraphics[width=\textwidth]{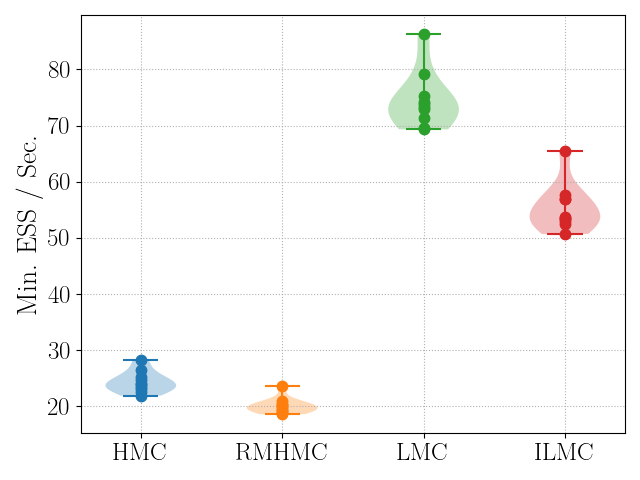}
  \caption{Min. ESS / Sec.}
  \end{subfigure}
  ~
  \begin{subfigure}[t]{0.3\textwidth}
  \includegraphics[width=\textwidth]{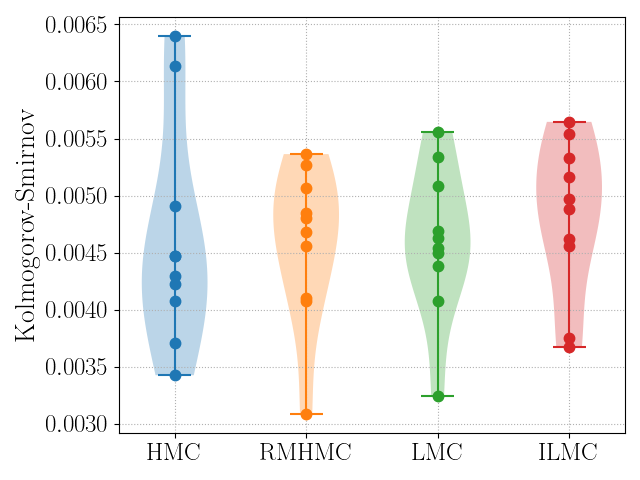}
  \caption{Kolmogorov-Smirnov}
  \end{subfigure}
  ~
  \begin{subfigure}[t]{0.3\textwidth}
  \includegraphics[width=\textwidth]{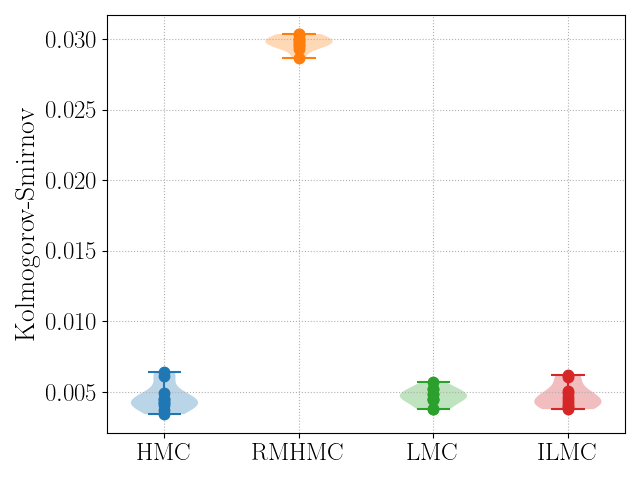}
  \caption{Kolmogorov-Smirnov (Misspecified)}
  \end{subfigure}
  \caption{{\bf (Fitzhugh-Nagumo Model)} We show the minimum ESS per second and the Kolmogorov-Smirnov ergodicity metric for inferring the posterior distribution of $(a, b,c)$ in the Fitzhugh-Nagumo posterior distribution. We observe that LMC has the strongest performance in terms of time-normalized ESS and that all methods exhibit comparable performance with respect to ergodicity. When we break the symmetry of partial derivatives, RMHMC no longer has the target distribution as its invariant distribution; this causes a degradation of the ergodicity metric. However, the geometric methods LMC and ILMC are robust to this misspecification because they are non-volume-preserving by design.}
  \label{lagrangian-remarks:fig:fitzhugh-nagumo}
\end{figure*}

We now investigate the LMC and RMHMC in the Fitzhugh-Nagumo posterior distribution. The Fitzhugh-Nagumo posterior distribution has the following generative model,
\begin{align}
    (a, b , c) &\overset{\mathrm{i.i.d.}}{\sim} \mathrm{Normal}(0, 1) \\
    \hat{r}_{t_k} &\overset{\mathrm{i.i.d.}}{\sim} \mathrm{Normal}(r_{t_k}, \sigma^2) \\
    \hat{v}_{t_k} &\overset{\mathrm{i.i.d.}}{\sim} \mathrm{Normal}(v_{t_k}, \sigma^2) ~~~\mathrm{for}~k=1,\ldots, n,
\end{align}
where $t_1,\ldots, t_n$ are evenly spaced time points in $[0, 10)$ and $r_t$ and $v_t$ obey the differential equations,
\begin{align}
  \dot{r}_t &= -\frac{v_t - a + b r_t}{c} \\
  \dot{v}_t &= cv_t - \frac{cv_t^3}{3} + cr_t.
\end{align}
In our experiments we give initial conditions $v_0 = 1$, $r_0 = -1$, $\sigma^2 = 1/4$, and $n=200$. The objective is to infer the posterior of $q=(a,b,c)$ given observations $\set{(\hat{v}_{t_k}, \hat{r}_{t_k})}_{k=1}^{200}$.

Here we emphasize a different aspect of the computation: its robustness to misspecification. As described in \cref{lagrangian-remarks:subsec:robustness}, the proof that the generalized leapfrog integrator is a volume-preserving transformation when used to integrate Hamiltonian vector fields is the symmetry of partial derivatives. We therefore consider the repercussions of an error in the implementation of partial derivatives that invalidates this requirement. To give further context to this experiment, we quote from \citet{dahlquist2003numerical}: ``In all numerical work, one must expect that clerical errors, errors in hand calculation, and misunderstandings will occur. [...] Most of the errors depend on the so-called human factor. [...] We  take up these sources of error in order to emphasize that both the person who carries out a calculation and the person who guides the work of others can plan so that such sources of error are not damaging.'' It seems to us, therefore, that the robustness of a Markov chain procedure to human misspecification is a most desirable circumstance.

As described in \citet{rmhmc}, computing the  gradient of the log-posterior, the Riemannian metric, and the Jacobian of the Riemannian metric requires us to compute sensitivity equations of the form $\frac{\mathrm{d}}{\mathrm{d}t} \frac{\partial r_t}{\partial q_i}$, $\frac{\mathrm{d}}{\mathrm{d}t} \frac{\partial v_t}{\partial q_i}$, $\frac{\mathrm{d}}{\mathrm{d}t} \frac{\partial^2 r_t}{\partial q_i \partial q_j}$, and $\frac{\mathrm{d}}{\mathrm{d}t} \frac{\partial^2 v_t}{\partial q_i \partial q_j}$ for $i=1,2, 3$ and $j=1,2,3$. We employ a Riemannian metric of the form,
\begin{align}
\mathbf{G}_{ij}(q) = \frac{1}{\sigma^2}\sum_{k=1}^{n} \paren{\frac{\partial v_{t_k}}{\partial q_i}\frac{\partial v_{t_k}}{\partial q_j} + \frac{\partial r_{t_k}}{\partial q_i}\frac{\partial r_{t_k}}{\partial q_j}} + \mathbf{1}\set{i=j}.
\end{align}
If any of these sensitivity equations are misspecified, then we will destroy the symmetry of partial derivatives required by the leapfrog integrator to preserve volume. However, the integrator used in LMC and ILMC is non-volume-preserving by design, and the appropriate volume correction is readily computable during sampling. We expect, therefore, that the volume correction employed in LMC and ILMC will provide robustness against these incorrectly computed quantities. We note that, in this case, there is no obvious mechanism by which to compute the required Jacobian determinant when the generalized leapfrog integrator is no longer a symplectic transformation.

In \cref{lagrangian-remarks:fig:fitzhugh-nagumo} we show the minimum ESS per second and Kolmogorov-Smirnov statistics for sampling from the Fitzhugh-Nagumo posterior. We observe that LMC outperforms ILMC in this example, but that both Lagrangian methods outperform HMC and RMHMC. In terms of ergodicity, all methods perform similarly. When we introduce changes to the sensitivity equations, we observe that RMHMC's ergodicity severely degrades. For HMC, for which higher-order sensitivities are not required (because of the Euclidean metric), ergodicity is identical to the previous case. However, for LMC and ILMC, in which the higher-order sensitivities are required to compute the Christoffel symbols, the fact that the sensitivities have been incorrectly specified has not noticeably degraded ergodicity. This is a virtue of the LMC and ILMC Markov chains that they are more robust to human errors of this variety.

\section{Conclusion}

This work has examined the numerical integrator for Lagrangian Monte Carlo (LMC). Motivated by the observation that LMC requires four Jacobian determinant evaluations, mechanisms by which this number may be reduced were examined. By inverting the sequence of integration so that position, rather than velocity, is updated twice, the number of Jacobian determinant evaluations in each step was reduced from four to two while still maintaining a fully explicit method. Empirical evaluations of this method were provided to show several situations in which the proposed integration strategy enjoys the best time-normalized performance among several alternatives. Moreover, it has been demonstrated in this work that the local error of the Lagrangian leapfrog (and its inverted counterpart) is third order, which improves the previously known order. Additionally, an important robustness property was characterized that LMC possesses and an illustration wherein human error will invalidate stationarity in RMHMC but not in LMC was given.

Methods of Bayesian inference that incorporate geometric understanding exhibit a pleasing aesthetic, yet they are burdened by numerical considerations -- among which are fixed points and cubic complexity -- that have limited their adoption. We hope that this research sparks renewed interest in mechanisms of improving these geometric methods.

\bigskip
\begin{center}
{\large\bf SUPPLEMENTARY MATERIAL}
\end{center}

\appendix

\section{Algorithms}\label{app:algorithms}

\begin{algorithm}[t!]
\caption{Pseudo-code implementing the generalized leapfrog integrator for computing approximate solutions to Hamilton's equations of motion. We assume a Hamiltonian in the form of \cref{lagrangian-remarks:eq:hamiltonian-form}.}
\label{lagrangian-remarks:alg:generalized-leapfrog}
\begin{algorithmic}[1]
  \STATE {\bf Input}: An integration step-size $\epsilon\in\R$, an initial position in phase space $(q, p)\in \R^m\times\R^m$.
  \STATE Compute an implicit half-step update to the momentum variable to obtain $\breve{p}$ using \cref{lagrangian-remarks:eq:generalized-leapfrog-momentum-i}.
  \STATE Compute an implicit full-step to the position variable to obtain $\tilde{q}$ using \cref{lagrangian-remarks:eq:generalized-leapfrog-position}.
  \STATE Compute an explicit half-step update to the momentum variable to obtain $\tilde{p}$ using \cref{lagrangian-remarks:eq:generalized-leapfrog-momentum-ii}.
  \STATE Compute the Jacobian determinant $J=1$.
  \STATE {\bf Output}: The updated position and momentum $(\tilde{q}, \tilde{p})$ and the Jacobian determinant $J$.
\end{algorithmic}
\end{algorithm}

\begin{algorithm}[t!]
\caption{Pseudo-code implementing the Lagrangian leapfrog integrator for computing approximate solutions to the Lagrangian equations of motion. We assume a Lagrangian in the form of \cref{lagrangian-remarks:eq:lagrangian-form}. The conversion from momentum to velocity and back again is included for notational simplicity; in practice, when applying multiple steps of this integrator, these conversions need only be computed for the first conversion from momentum to velocity and the last conversion from velocity to momentum.}
\label{lagrangian-remarks:alg:lagrangian-leapfrog}
\begin{algorithmic}[1]
  \STATE {\bf Input}: An integration step-size $\epsilon\in\R$, an initial position in phase space $(q, p)\in \R^m\times\R^m$.
  \STATE Use the inverse Legendre transform (\cref{lagrangian-remarks:def:inverse-legendre-transform}) to convert from momentum to velocity $v = \mathbf{G}^{-1}(q) p$.
  \STATE Compute an explicit half-step update to the velocity variable to obtain $\breve{v}$ using \cref{lagrangian-remarks:eq:lagrangian-velocity-i}.
  \STATE Compute an explicit full-step to the position variable to obtain $\tilde{q}$ using \cref{lagrangian-remarks:eq:lagrangian-position-update}.
  \STATE Compute an explicit half-step update to the velocity variable to obtain $\tilde{v}$ using \cref{lagrangian-remarks:eq:lagrangian-velocity-ii}.
  \STATE Use the Legendre transform (\cref{lagrangian-remarks:def:legendre-transform}) to convert from velocity to momentum $\tilde{p} = \mathbf{G}(\tilde{q}) \tilde{v}$.
  \STATE Compute the Jacobian determinant $J = \frac{\partial (\tilde{q},\tilde{p})}{\partial (q, p)}$ using \cref{lagrangian-remarks:eq:lagrangian-jacobian-determinant}.
  \STATE {\bf Output}: The updated position and momentum $(\tilde{q}, \tilde{p})$ and the Jacobian determinant $J$.
\end{algorithmic}
\end{algorithm}

\begin{algorithm}[t!]
\caption{Single step of the Hamiltonian / Lagrangian Monte Carlo transition kernel.}
\label{lagrangian-remarks:alg:hmc}
\begin{algorithmic}[1]
  \STATE {\bf Input}: The current state of the Markov chain $(q_n, p_n)$, an integration step-size $\epsilon\in\R$, a number of integration steps $k\in\mathbb{N}$, a numerical integrator $\Phi$ (as described in any of \cref{lagrangian-remarks:alg:generalized-leapfrog,lagrangian-remarks:alg:lagrangian-leapfrog,lagrangian-remarks:alg:inverted-lagrangian-leapfrog}).
  \STATE Resample the momentum $p_n\vert q_n \sim \mathrm{Normal}(0, \mathbf{G}(q_n))$.
  \STATE Initialize the Jacobian determinant $J=1$ and initial integration state $(\tilde{q},\tilde{p}) = (q_n, p_n)$.
  \FOR{$i = 1,\ldots, k$}
  \STATE Compute the proposal by applying the numerical integrator: $((\tilde{q}, \tilde{p}), \tilde{J}) = \Phi_{\epsilon}(\tilde{q}, \tilde{p})$.
  \STATE Update the Jacobian determinant $J = J \times \tilde{J}$
  \ENDFOR
  \STATE Apply the momentum flip operator (\cref{lagrangian-remarks:def:momentum-flip}) to obtain $(\tilde{q}, \tilde{p}) = \mathbf{F}(\tilde{q}, \tilde{p})$.
  \STATE Sample $u\sim\mathrm{Uniform}(0, 1)$ and compute the acceptance probability $a = \alpha((q_n, p_n), (\tilde{q}, \tilde{p}), J)$ using \cref{lagrangian-remarks:eq:acceptance-probability}.
  \IF{$u < a$}
  \STATE Accept the proposal $(q_{n+1}, p_{n+1}) = (\tilde{q}, \tilde{p})$.
  \ELSE
  \STATE Reject the proposal $(q_{n+1}, p_{n+1}) = (q_n, p_n)$.
  \ENDIF
  \STATE {\bf Output}: The next state of the Markov chain $(q_{n+1}, p_{n+1})$.
\end{algorithmic}
\end{algorithm}

\begin{algorithm}[t!]
\caption{Pseudo-code implementing the inverted Lagrangian leapfrog integrator for computing approximate solutions to the Lagrangian equations of motion. We assume a Lagrangian in the form of \cref{lagrangian-remarks:eq:lagrangian-form}. This version of leapfrog integration for Lagrangian dynamics requires only two Jacobian determinant computations in each step, as opposed to the four required by \cref{lagrangian-remarks:alg:lagrangian-leapfrog}.}
\label{lagrangian-remarks:alg:inverted-lagrangian-leapfrog}
\begin{algorithmic}[1]
  \STATE {\bf Input}: An integration step-size $\epsilon\in\R$, an initial position in phase space $(q, p)\in \R^m\times\R^m$.
  \STATE Use the inverse Legendre transform (\cref{lagrangian-remarks:def:inverse-legendre-transform}) to convert from momentum to velocity $v = \mathbf{G}^{-1}(q) p$.
  \STATE Compute an explicit half-step to the position variable to obtain $\breve{v}$ using \cref{lagrangian-remarks:eq:inverted-lagrangian-position-i}.
  \STATE Compute an explicit full-step update to the velocity variable to obtain $\tilde{v}$ using \cref{lagrangian-remarks:eq:inverted-lagrangian-velocity}.
  \STATE Compute an explicit half-step to the position variable to obtain $\tilde{q}$ using \cref{lagrangian-remarks:eq:inverted-lagrangian-position-ii}.
  \STATE Use the Legendre transform (\cref{lagrangian-remarks:def:legendre-transform}) to convert from velocity to momentum $\tilde{p} = \mathbf{G}(\tilde{q}) \tilde{v}$.
  \STATE Compute the Jacobian determinant $J = \frac{\partial (\tilde{q},\tilde{p})}{\partial (q, p)}$ using \cref{lagrangian-remarks:eq:lagrangian-jacobian-determinant}.
  \STATE {\bf Output}: The updated position and momentum $(\tilde{q}, \tilde{p})$ and the Jacobian determinant $J$.
\end{algorithmic}
\end{algorithm}

In this appendix we give pseudo-code implementations of algorithms featured in this work. In \cref{lagrangian-remarks:alg:generalized-leapfrog} we show the generalized leapfrog, which contrasts with the Lagrangian leapfrog in \cref{lagrangian-remarks:alg:lagrangian-leapfrog} in the presence of implicitly-defined integration steps. \Cref{lagrangian-remarks:alg:inverted-lagrangian-leapfrog} shows the Lagrangian integrator but with an inverted sequence of integration (updating position before velocity). In \cref{lagrangian-remarks:alg:hmc} we show an algorithmic implementation of a single-step of the Hamiltonian Monte Carlo Markov chain, which can be carried out using any of the integrators considered in this work.

\section{Expanded Preliminaries}

\subsection{Hamiltonian and Lagrangian Mechanics}\label{lagrangian-remarks:subsec:hamiltonian-lagrangian-mechanics}

\begin{definition}\label{lagrangian-remarks:def:hamiltonian-mechanics}
Let $H : \R^m\times\R^m\to\R$ be a smooth function, which we call the Hamiltonian. Hamilton's equations of motion are defined as the solutions to the initial value problem
\begin{align}
  \label{lagrangian-remarks:eq:hamiltonian-position} v_t &= \nabla_p H(q_t, p_t) \\
  \label{lagrangian-remarks:eq:hamiltonian-momentum} \dot{p}_t &= -\nabla_q H(q_t, p_t),
\end{align}
with $(q_0, p_0)$ a given initial position in phase space.
\end{definition}
\Cref{lagrangian-remarks:def:hamiltonian-mechanics} gives us a system of coupled first-order differential equations. Hamilton's equations of motion exhibit several key properties \citep{mechanics-and-symmetry}, which we summarize. 
\begin{theorem}\label{lagrangian-remarks:thm:hamiltonian-mechanics}
The Hamiltonian mechanics given in \cref{lagrangian-remarks:def:hamiltonian-mechanics} possess the following three properties:
\begin{enumerate}
    \item They preserve the Hamiltonian $\frac{\mathrm{d}}{\mathrm{d}t} H(q_t, p_t) = 0$.
    \item Denoting $\dot{z}_t = (v_t,\dot{p}_t)$ the equations of motion in phase space, Hamiltonian mechanics preserve volume in phase space: $\mathrm{div}(\dot{z}_t) = 0$.
    \item Under the conditions that $\nabla_q H(q, -p) = \nabla_q H(q, p)$ and $-\nabla_p H(q,p)= \nabla_pH(q, -p)$, the equations of motion are reversible via negation of the momentum variable.
\end{enumerate}
\end{theorem}
A proof is given in \cref{lagrangian-remarks:app:proofs-concerning-mechanics}.
In the special case when $U(q) = 0$ in \cref{lagrangian-remarks:eq:hamiltonian-form} for all $q\in\R^m$, Hamilton's equations of motion produce geodesic motion on the Riemannian manifold $(\R^m, \mathbf{G})$, where $\mathbf{G}$ assumes the role of the Riemannian metric \citep{Calin2004GeometricMO}.

\begin{definition}
  Given a Riemannian Hamiltonian (\cref{lagrangian-remarks:def:riemannian-hamiltonian}), define the associated {\it Lagrangian} function by,
  \begin{align}
    \label{lagrangian-remarks:eq:lagrangian-form} L(q, v) &= \tilde{K}(q, v) - U(q),
  \end{align} 
  where $\tilde{K}(q, v) = \frac{1}{2}v^\top \mathbf{G}(q) v$.
\end{definition}
Hamiltonian mechanics are related to Lagrangian dynamics by the Legendre transform which converts between momentum and velocity \citep{marsden_west_2001}.
\begin{definition}\label{lagrangian-remarks:def:legendre-transform}
  Let $L : \R^m\times\R^m\to\R$ be a Lagrangian of the form in \cref{lagrangian-remarks:eq:lagrangian-form}. The {\it Legendre transform} relates the momentum and velocity according to $p = \frac{\partial L}{\partial v}(q, v) = \mathbf{G}(q) v$. 
\end{definition}
\begin{definition}\label{lagrangian-remarks:def:inverse-legendre-transform}
  Let $H:\R^m\times\R^m\to\R$ be a Hamiltonian of the form in \cref{lagrangian-remarks:eq:hamiltonian-form}. The {\it inverse Legendre transform} relates velocity and momentum according to $v = \frac{\partial H}{\partial p}(q, p) = \mathbf{G}^{-1}(q)p$.
\end{definition}
The Lagrangian then determines equations of motion in accordance with the following physical principle.
\begin{definition}\label{lagrangian-remarks:def:hamiltons-principle}
{\it Hamilton's principle} states that the equations of motion $(q_t, v_t)$ over an interval of time $[a, b]$, with known boundary conditions $q_a$ and $q_b$, are solutions of the variational equation
\begin{align}
    \frac{\delta \mathcal{S}[q_t]}{\delta q_t} = 0
\end{align}
where
\begin{align}
    \mathcal{S}[q_t] = \int_a^b L(q_t, v_t) ~\mathrm{d}t
\end{align}
where $q_t$ and $v_t$ are related by $v_t = \frac{\mathrm{d}}{\mathrm{d}t} q_t$.
\end{definition}
Hamilton's principle states that equations of motion (as specified by $q$ and $v$) should {\it extremize} the Lagrangian subject to the boundary conditions on $q_a$ and $q_b$. 
\begin{theorem}\label{lagrangian-remarks:thm:euler-lagrange}
Under Hamilton's principle (\cref{lagrangian-remarks:def:hamiltons-principle}), the equations of motion $(q_t, v_t)$ must be solutions of the {\it Euler-Lagrange equation}
\begin{align}
  \nabla_q L(q_t, v_t) = \frac{\mathrm{d}}{\mathrm{d}t} \nabla_{v} L(q_t, v_t).
\end{align}
\end{theorem}
The equations of motion from \cref{lagrangian-remarks:thm:euler-lagrange} (or \cref{lagrangian-remarks:def:hamiltons-principle}) are, in fact, equivalent to the motion produced under Hamiltonian mechanics in \cref{lagrangian-remarks:def:hamiltonian-mechanics} \citep{mechanics-and-symmetry}. For Lagrangians of the form in \cref{lagrangian-remarks:eq:lagrangian-form}, the $k$-th element of the acceleration is,
\begin{align}
  \label{lagrangian-remarks:eq:lagrangian-acceleration} a_t^{(k)} = -\sum_{i=1}^m\sum_{j=1}^m \Gamma^k_{ij}(q_t) v_t^{(i)} v_t^{(j)} - \sum_{l=1}^m \mathbf{G}^{-1}_{kl}(q_t) \frac{\partial U}{\partial q^{(l)}} (q_t),
\end{align}
where $\Gamma^k_{ij}$ are the Christoffel symbols (\cref{lagrangian-remarks:def:christoffel-symbols}).
We conclude this section by noting that $\Omega$ in \cref{lagrangian-remarks:eq:omega} enjoys the following properties:
\begin{proposition}\label{lagrangian-remarks:prop:omega-properties}
The function $\Omega : \R\times\R^m\times\R^m\to\R^{m\times m}$ in \cref{lagrangian-remarks:eq:omega} satisfies 
  \begin{align}
      \label{lagrangian-remarks:eq:omega-symmetry} \Omega(\epsilon, q, v)\breve{v} &= \Omega(\epsilon, q, \breve{v}) v \\
      \label{lagrangian-remarks:eq:omega-derivative} \frac{\partial}{\partial v} (\Omega(\epsilon, q, v)\breve{v}) &= \Omega(\epsilon, q, \breve{v}).
  \end{align}
\end{proposition}
\begin{proof}
The result in \cref{lagrangian-remarks:eq:omega-symmetry} can be seen from
\begin{align}
    (\Omega(\epsilon, q, v)\breve{v})_i &= \frac{\epsilon}{2} \sum_{j=1}^m \sum_{k=1}^m \Gamma^i_{kj}(q) v^{(k)} \breve{v}^{(j)} \\
    &= \frac{\epsilon}{2} \sum_{j=1}^m \sum_{k=1}^m \Gamma^i_{jk}(q) v^{(k)} \breve{v}^{(j)} \\
    &= \frac{\epsilon}{2} \sum_{k=1}^m \sum_{j=1}^m \Gamma^i_{jk}(q) v^{(k)} \breve{v}^{(j)} \\
    &= (\Omega(\epsilon, q, \breve{v})v)_i.
\end{align}
\Cref{lagrangian-remarks:eq:omega-derivative} then follows as an immediate corollary.
\end{proof}
Like Hamiltonian mechanics, Lagrangian dynamics conserve the Hamiltonian (when $p_t = \mathbf{G}(q_t) v_t$) and are reversible. However, Lagrangian dynamics do not conserve volume in $(q, v)$-space.

\subsection{Background on Numerical Integrators}\label{lagrangian-remarks:subsec:numerical-integrators}

The Hamiltonian equations of motion in \cref{lagrangian-remarks:eq:hamiltonian-position,lagrangian-remarks:eq:hamiltonian-momentum} and the Lagrangian motion described in \cref{lagrangian-remarks:eq:lagrangian-acceleration} rarely have closed-form solutions. Therefore, it is necessary to investigate methods of numerical integration to produce approximate solutions to these initial value problems. In the following section, we review key ideas from \citet{Hairer:1250576}.
\begin{definition}\label{lagrangian-remarks:def:initial-value-problem}
  Let $g : \R^m\to\R^m$. {\it The solution to an initial value problem} is a function $z_{(\cdot)} : \R\to\R^m$ for which $\frac{\mathrm{d}}{\mathrm{d}t} z_t = g(z_t)$ and for which the {\it initial value} $z_0\in\R^m$ is known. In this case $g$ is called a {\it (time-homogenous) vector field}.
\end{definition}
\begin{definition}\label{lagrangian-remarks:def:integrator-order}
  Let $z_t$ be the solution to an initial value problem (\cref{lagrangian-remarks:def:initial-value-problem}) with initial value $z_0$. A {\it numerical integrator with step-size $\epsilon$} is a mapping $\Phi_{\epsilon} : \R^m\to\R^m$ and is said to have order $p$ if $\Vert \Phi_{\epsilon}(z_0) - z_{\epsilon}\Vert = \mathcal{O}(\epsilon^{p+1})$.
\end{definition}
\begin{definition}\label{lagrangian-remarks:def:self-adjoint}
  The adjoint of a numerical method $\Phi_\epsilon : \R^m \to\R^m$ is defined by the relation $\Phi_\epsilon^*(z) = \Phi^{-1}_{-\epsilon}(z)$. The numerical method $\Phi_\epsilon$ is said to be self-adjoint if $\Phi^*_\epsilon = \Phi_\epsilon$.
\end{definition}
\begin{theorem}[\citet{Hairer:1250576}]\label{lagrangian-remarks:thm:self-adjoint-even-order}
  Consider an initial value problem $\dot{z}_t = g(z_t)$ with initial condition $z_0 \in \R^m$. Let $\Phi_\epsilon$ be a one-step numerical integrator of (maximal) order $r\in\mathbb{N}$. If $\Phi_\epsilon$ is self-adjoint, then $r$ is even.
\end{theorem}
\begin{proposition}\label{lagrangian-remarks:prop:composition-order}
Let $g : \R^m\to\R^m$ be a time-homogenous vector field. Let $\Xi:\R^m\to\R^m$ be a diffeomorphism. Let $z_t$ be the solution to the initial value problem $\dot{z}_t = g(z_t)$ given $z_0$. Let $\Phi_\epsilon : \R^m\to\R^m$ be the flow map of $z_\epsilon$ and suppose that $\hat{\Phi}_\epsilon$ is a $p$-th order approximation of $\Phi_\epsilon$. Then $\Xi\circ \hat{\Phi}_\epsilon$ is a $p$-th order approximation of $\Xi\circ \Phi_\epsilon$.
\end{proposition}
A proof is given in \cref{lagrangian-remarks:app:proofs-concerning-numerical-order}.

\begin{proposition}
  The generalized leapfrog integrator is a second-order (\cref{lagrangian-remarks:def:integrator-order}), self-adjoint (\cref{lagrangian-remarks:def:self-adjoint}) numerical method.
\end{proposition}

\subsection{Hamiltonian and Lagrangian Monte Carlo}\label{lagrangian-remarks:subsec:hamiltonian-lagrangian-monte-carlo}

Our objective in Bayesian inference is to draw samples from the distribution whose density is $\pi(q) \propto \exp(\mathcal{L}(q))$. We now review basic concepts from Markov chain Monte Carlo.
\begin{definition}
  A {\it Markov chain transition kernel} is a map $K : \R^m\times\mathfrak{B}(\R^m) \to [0,1]$ satisfying (i) for every $x\in \R^m$ the map $A\mapsto K(x, A)$ is a probability measure and (ii) for every $A \in\mathfrak{B}(\R^m)$, the map $x\mapsto K(x, A)$ is measurable. Given a Markov chain transition kernel $K$, a {\it Markov chain} is defined inductively by $x_{n+1} \vert x_n \sim K(x_n, \cdot)$.
\end{definition}
HMC accomplishes this by artificially expanding the distribution to incorporate a momentum variable. Defining $U(q) = -\mathcal{L}(q) + \frac{1}{2} \log\mathrm{det}(\mathbf{G}(q))$, let $H(q, p)$ be the Hamiltonian given in \cref{lagrangian-remarks:eq:hamiltonian-form} and observe that the density $\pi(q, p) \propto \exp(-H(q, p))$ has $\pi(q) = \int_{\R^m} \pi(q, p) ~\mathrm{d}p$ and $p \vert q \sim\mathrm{Normal}(0, \mathbf{G}(q))$. In order to unify methods of Bayesian inference based on either the generalized leapfrog integrator (\cref{lagrangian-remarks:def:generalized-leapfrog}) or the Lagrangian leapfrog (\cref{lagrangian-remarks:def:lagrangian-leapfrog}) under one framework, we now introduce the Markov chain transition kernel based on smooth involutions.
\begin{definition}[\citet{DBLP:conf/icml/NeklyudovWEV20}]\label{lagrangian-remarks:def:involutive-monte-carlo}
  Let $\Phi : \R^m\times\R^m\to\R^m\times\R^m$ be a smooth involution (i.e. $\Phi=\Phi^{-1}$). Let $\pi:\R^m\times\R^m\to\R_+$ be a probability density with respect to Lebesgue measure. Then we define the Markov chain transition kernel of {\it involutive Monte Carlo} by
  \begin{align}
  \begin{split}
      K((q, p), (A, B)) &= \alpha((q, p), (\tilde{q},\tilde{p}), J)~\mathbf{1}\set{(\tilde{q}, \tilde{p})\in (A, B)} \\
      &\qquad + \paren{1-\alpha((q, p), (\tilde{q},\tilde{p}), J)} ~\mathbf{1}\set{(q, p)\in (A, B)}
  \end{split} \\
  \label{lagrangian-remarks:eq:acceptance-probability} \alpha((q, p), (\tilde{q}, \tilde{p}), J) &= \min\set{1, \frac{\exp(-H(\tilde{q}, \tilde{p}))}{\exp(-H(q, p))}\cdot J} \\
  J &= \abs{\mathrm{det}\paren{\frac{\partial (\tilde{q}, \tilde{p})}{\partial (q, p)}}},
  \end{align}
  where $(A, B)\in\mathfrak{B}(\R^m\times\R^m)$ and $(\tilde{q}, \tilde{p}) = \Phi(q, p)$.
\end{definition}
\begin{proposition}\label{lagrangian-remarks:prop:involutive-monte-carlo-detailed-balance}
  The Markov chain transition kernel of involutive Monte Carlo satisfies detailed balance with respect to the distribution whose density is $\pi(q, p)\propto \exp(-H(q, p))$.
\end{proposition}
A proof is provided in \cref{lagrangian-remarks:app:involutive-monte-carlo}. Central to the construction of involutions of interest to us is the momentum flip operator, defined as follows.
\begin{definition}\label{lagrangian-remarks:def:momentum-flip}
  The {\it momentum flip operator} is the map $\mathbf{F}(q, p) = (q, -p)$.
\end{definition}
The fact that numerical integrators can be combined with the momentum flip operator in order to produce involutions is covered in \cref{lagrangian-remarks:app:involutive-monte-carlo} in the case of the Lagrangian leapfrog; other integrators are handled similarly.
We provide pseudo-code implementing a single step of involutive Monte Carlo in \cref{lagrangian-remarks:alg:hmc} with involutions provided by the generalized leapfrog or Lagrangian leapfrog integrator. Given an initial point $(q_0, p_0)$ in phase space drawn from the distribution $\pi(q, p)$, the sequence of states $((q_1,p_1), (q_2, p_2), \ldots)$ computed by \cref{lagrangian-remarks:alg:hmc} are guaranteed to have $\pi(q, p)$ as their marginal distributions. Under the additional conditions that the HMC Markov chain is irreducible and aperiodic, HMC also produces an ergodic chain. 

The generalized leapfrog integrator is volume preserving. However, the integrator of Lagrangian dynamics is not. The required change-of-volume can be deduced as follows. First, observe that the update \cref{lagrangian-remarks:eq:lagrangian-position-update} is immediately volume-preserving since it is merely the translation of the position variable by a quantity. The change of volume incurred in \cref{lagrangian-remarks:eq:lagrangian-velocity-i} has a Jacobian determinant given by,
\begin{align}
  \abs{\mathrm{det}\paren{\frac{\partial (q, \breve{v})}{\partial (q, v)}}} = \abs{\frac{\mathrm{det}(\mathrm{Id}_m - \Omega(\epsilon, q, \breve{v}))}{\mathrm{det}(\mathrm{Id}_m + \Omega(\epsilon, q, v))}}.
\end{align}
The update in \cref{lagrangian-remarks:eq:lagrangian-velocity-ii} incurs a similar change of volume. Thus, when employing \cref{lagrangian-remarks:alg:lagrangian-leapfrog}, the Jacobian determinant of the transformation $(q, p)\mapsto (\tilde{q}, \tilde{p})$ is,
\begin{align}
  \begin{split}
    \label{lagrangian-remarks:eq:lagrangian-jacobian-determinant} &\abs{\mathrm{det}\paren{\frac{\partial (\tilde{q}, \tilde{p})}{\partial (q, p)}}} = \mathrm{det}(\mathbf{G}^{-1}(q))\mathrm{det}(\mathbf{G}(\tilde{q})) ~\times \\
    &\qquad \abs{\frac{\mathrm{det}(\mathrm{Id}_m - \Omega(\epsilon, \tilde{q}, \tilde{v})) \mathrm{det}(\mathrm{Id}_m - \Omega(\epsilon, q, \breve{v}))}{\mathrm{det}(\mathrm{Id}_m + \Omega(\epsilon, \tilde{q}, \breve{v})) \mathrm{det}(\mathrm{Id}_m + \Omega(\epsilon, q, v))}}.
  \end{split}
\end{align}
\begin{definition}\label{lagrangian-remarks:def:rmhmc}
  Let $\Phi_\epsilon$ be the generalized leapfrog integrator with step-size $\epsilon\in\R$ (\cref{lagrangian-remarks:def:generalized-leapfrog}). Let $k\in\mathbb{N}$ be a number of integration steps. The {\it Riemannian manifold Hamiltonian Monte Carlo (RMHMC) Markov chain} is an instance involutive Monte Carlo (\cref{lagrangian-remarks:def:involutive-monte-carlo}) with involution $\mathbf{F} \circ \Phi_{\epsilon}^k$.
\end{definition}
\begin{definition}\label{lagrangian-remarks:def:lmc}
  Let $\Phi_\epsilon$ be the Lagrangian leapfrog integrator with step-size $\epsilon\in\R$ (\cref{lagrangian-remarks:def:lagrangian-leapfrog}). Let $k\in\mathbb{N}$ be a number of integration steps. The {\it Lagrangian Monte Carlo (LMC) Markov chain} is an instance involutive Monte Carlo (\cref{lagrangian-remarks:def:involutive-monte-carlo}) with involution $\mathbf{F} \circ \Phi_{\epsilon}^k$.
\end{definition}

\section{Proofs Concerning Mechanics Systems}\label{lagrangian-remarks:app:proofs-concerning-mechanics}

We give a proof of \cref{lagrangian-remarks:thm:hamiltonian-mechanics}.
\begin{proof}
We first show that the Hamiltonian energy is conserved. 
\begin{align}
    \frac{\mathrm{d}}{\mathrm{d}t} H(q_t, p_t) &= \nabla_q H(q_t, p_t) \cdot v_t + \nabla_p H(q_t, p_t) \cdot \dot{p}_t \\
    &= \nabla_q H(q_t, p_t) \cdot \nabla_p H(q_t, p_t) - \nabla_p H(q_t, p_t) \cdot \nabla_q H(q_t, p_t) \\
    &= 0.
\end{align}
Next we show that Hamiltonian mechanics conserve volume in $(q, p)$-space. This is equivalent to the vector field $(v_t, \dot{p}_t)$ having zero divergence, which we now verify.
\begin{align}
    \mathrm{div}(v_t, \dot{p}_t) &= \nabla_q \cdot v_t + \nabla_p \cdot \dot{p}_t \\
    &= \nabla_q \cdot \nabla_p H(q_t, p_t) - \nabla_p \cdot \nabla_q H(q_t, p_t) \\
    &= \sum_{i=1}^m \frac{\partial^2 H}{\partial q_i \partial p_i} H(q_t, p_t) - \sum_{i=1}^m \frac{\partial^2 H}{\partial p_i \partial q_i} H(q_t, p_t) \\
    &= 0,
\end{align}
by symmetry of partial derivatives. Finally we show that under the conditions that $\nabla_q H(q, -p) = \nabla_q H(q, p)$ and $\nabla_p H(q, -p) = -\nabla_p H(q, p)$ that the equations of motion are reversible under negation of the momentum variable. To see this, fix $\tau \in \R_+$ and consider $(q_t, p_t)$ satisfying \cref{lagrangian-remarks:eq:hamiltonian-position,lagrangian-remarks:eq:hamiltonian-momentum} for $t\in [0, \tau]$. Let $\tilde{p}_t = -p_{\tau-t}$ and $\tilde{q}_t = q_{\tau-t}$. We find that $\tilde{q}_t$ and $\tilde{p}_t$ obey the following equations of motion:
\begin{align}
    \frac{\mathrm{d}}{\mathrm{d}t} \tilde{q}_t &= -v_{\tau-t} \\
    &= -\nabla_p H(q_{\tau-t}, p_{\tau-t}) \\
    &= \nabla_p H(\tilde{q}_t, -p_{\tau-t}) \\
    &= \nabla_p H(\tilde{q}_t, \tilde{p}_{t}) \\
    \frac{\mathrm{d}}{\mathrm{d}t} \tilde{p}_t &= \dot{p}_{\tau-t} \\
    &= -\nabla_q H(q_{\tau-t}, p_{\tau-t}) \\
    &= -\nabla_q H(q_{\tau-t}, -p_{\tau-t}) \\
    &= -\nabla_q H(\tilde{q}_t, \tilde{p}_t).
\end{align}
Thus we see that $(\tilde{q}_t, \tilde{p}_t)$ are also solutions to Hamilton's equations of motion and satisfy $\tilde{q}_\tau = q_0$ and $\tilde{p}_\tau = -p_0$, demonstrating reversibility.
\end{proof}
\clearpage
\section{Proofs Concerning the Numerical Order of the Lagrangian Integrators}\label{lagrangian-remarks:app:proofs-concerning-numerical-order}

The claim of first order accuracy was derived in \citet{lan2015}; we have included a proof for completeness.
\begin{proof}[Proof of \cref{lagrangian-remarks:lem:second-order-local-error}]
Let $q_t$ and $v_t$ be solutions to the initial value problem,
\begin{align}
    \dot{q}_t &= v_t \\
    \dot{v}_t &= -\Omega(2, q_t, v_t) v_t - \mathbf{G}^{-1}(q_t) \nabla U(q_t).
\end{align}
Given the initial value $q_0 = q$ and $v_0 = v$, we can expand the solution in a Taylor series about $t=0$.
\begin{align}
    q_t &= q_0 + t \dot{q}_0 + \mathcal{O}(t^2) \\
    &= q + t v + \mathcal{O}(t^2) \\
    v_t &= v_0 - t \paren{\Omega(2, q_0, v_0) v_0 - \mathbf{G}^{-1}(q_0) \nabla U(q_0)} + \mathcal{O}(t^2) \\
    &= v - \Omega(2t, q, v) v - t \mathbf{G}^{-1}(q) \nabla U(q) + \mathcal{O}(t^2).
\end{align}
Letting $t=\epsilon$ gives the following approximations,
\begin{align}
    q_{\epsilon} &= q + \epsilon v + \mathcal{O}(\epsilon^2) \\
    v_{\epsilon} &= v - \Omega(2\epsilon, q, v) v - \epsilon \mathbf{G}^{-1}(q) \nabla U(q) + \mathcal{O}(\epsilon^2)
\end{align}
Now we expand the steps of the explicit numerical integrator. We begin with the first update to velocity.
\begin{align}
\breve{v} &= \paren{\mathrm{Id}_m + \Omega(\epsilon, q, v)}^{-1} \paren{v - \frac{\epsilon}{2} \mathbf{G}^{-1}(q)\nabla U(q)} \\
&= \paren{\mathrm{Id}_m - \Omega(\epsilon, q, v) + \mathcal{O}(\epsilon^2)}\paren{v - \frac{\epsilon}{2} \mathbf{G}^{-1}(q)\nabla U(q)} \\
&= v - \Omega(\epsilon, q, v) v - \frac{\epsilon}{2} \mathbf{G}^{-1}(q) \nabla U(q) + \mathcal{O}(\epsilon^2)
\end{align}
Now we expand the update to position.
\begin{align}
\tilde{q} &= q + \epsilon \breve{v} \\
&= q + \epsilon v + \mathcal{O}(\epsilon^2)
\end{align}
We conclude by expanding the second update to velocity in terms of the initial conditions.
\begin{align}
\tilde{v} &= \paren{\mathrm{Id}_m + \Omega(\epsilon, \tilde{q}, \breve{v})}^{-1} \paren{\breve{v} - \frac{\epsilon}{2} \mathbf{G}^{-1}(\tilde{q})\nabla U(\tilde{q})} \\
&= \breve{v} - \Omega(\epsilon, \tilde{q}, \breve{v}) \breve{v} - \frac{\epsilon}{2} \mathbf{G}^{-1}(\tilde{q}) \nabla U(\tilde{q}) + \mathcal{O}(\epsilon^2) \\
&= \breve{v} - \Omega(\epsilon, q, v) v - \frac{\epsilon}{2} \mathbf{G}^{-1}(q) \nabla U(q) + \mathcal{O}(\epsilon^2) \\
&= v - \Omega(2\epsilon, q, v) v - \epsilon \mathbf{G}^{-1}(q) \nabla U(q) + \mathcal{O}(\epsilon^2).
\end{align}
Therefore, a single step of the numerical integrator with step-size $\epsilon$ agrees with the analytical solution to the initial value problem (from the same initial condition) to at least first order in $\epsilon$.
\end{proof}

\begin{proof}[Proof of \Cref{lagrangian-remarks:lem:inverted-lagrangian-properties}]
To demonstrate that the inverted integrator has at least first-order error, we expand the steps of the integrator as follows. First,
\begin{align}
    \breve{q} &= q + \frac{\epsilon}{2} ~v \\
    \tilde{v} &= \paren{\mathrm{Id}_m + \Omega(2\epsilon, \breve{q}, v)}^{-1} \paren{v - \epsilon \mathbf{G}^{-1}(\breve{q})\nabla U(\breve{q})} \\
    &= \paren{\mathrm{Id}_m + \Omega(2\epsilon, \breve{q}, v)} \paren{v - \epsilon \mathbf{G}^{-1}(\breve{q})\nabla U(\breve{q})} + \mathcal{O}(\epsilon^2) \\
    &= v - \Omega(2\epsilon, q, v) v - \epsilon \mathbf{G}^{-1}(q) \nabla U(q) + \mathcal{O}(\epsilon^2) \\
    \tilde{q} &= \breve{q} + \frac{\epsilon}{2} \tilde{v} \\
    &= q + \frac{\epsilon}{2} ~v + \frac{\epsilon}{2} ~v + \mathcal{O}(\epsilon^2) \\
    &= q + \epsilon ~v.
\end{align}
This verifies that the inverted Lagrangian leapfrog has at least first order accuracy. 

In order to show that the inverted leapfrog integrator is symmetric, we proceed as follows. Recall that the three steps of the inverted Lagrangian leapfrog are
\begin{align}
    \breve{q} &= q + \frac{\epsilon}{2} ~v \\
    \tilde{v} &= \paren{\mathrm{Id}_m + \Omega(2\epsilon, \breve{q}, v)}^{-1} \paren{v - \epsilon \mathbf{G}^{-1}(\breve{q})\nabla U(\breve{q})} \\
    \tilde{q} &= \breve{q} + \frac{\epsilon}{2} ~\tilde{v}.
\end{align}
Therefore, we consider integrating from initial position $(\tilde{q},\tilde{v})$ with a negated step-size $-\epsilon$ as follows:
\begin{align}
    \breve{q}' &= \tilde{q} - \frac{\epsilon}{2} ~\tilde{v} \\
    &= \breve{q}
\end{align}
For the velocity we have,
\begin{align}
    & \tilde{v}' = \paren{\mathrm{Id}_m - \Omega(2\epsilon, \breve{q}, \tilde{v})}^{-1} \paren{\tilde{v} + \epsilon \mathbf{G}^{-1}(\breve{q})\nabla U(\breve{q})} \\
    \implies& \paren{\mathrm{Id}_m - \Omega(2\epsilon, \breve{q}, \tilde{v})} \tilde{v}' = \tilde{v} + \epsilon \mathbf{G}^{-1}(\breve{q})\nabla U(\breve{q}) \\
    \implies& \tilde{v} + \Omega(2\epsilon, \breve{q}, \tilde{v}')\tilde{v} = \tilde{v}' - \epsilon \mathbf{G}^{-1}(\breve{q})\nabla U(\breve{q}) \\
    \implies& \tilde{v} = \paren{\mathrm{Id}_m + \Omega(2\epsilon, \breve{q}, \tilde{v}'}^{-1} \paren{\tilde{v}' - \epsilon \mathbf{G}^{-1}(\breve{q})\nabla U(\breve{q})} \\
    \implies \tilde{v}' = v
\end{align}
Finally, the last update to the position is,
\begin{align}
    \tilde{q}' &= \breve{q} - \frac{\epsilon}{2} v \\
    &= q.
\end{align}
Hence we see that the inverted Lagrangian leapfrog is also self-adjoint. As noted in the main text, the combination of self-adjointness and at least first-order accuracy immediately imply second-order accuracy.
\end{proof}

\begin{proof}[Proof of \cref{lagrangian-remarks:lem:self-adjoint}]
Self-adjointness of a numerical integrator follows immediately from the condition $\Phi_{-\epsilon}\circ \Phi_{\epsilon} = \mathrm{Id}$. Therefore, to demonstrate that a numerical method is self-adjoint it suffices to establish this condition. Consider the first update to the velocity:
\begin{align}
    &\breve{v} = \paren{\mathrm{Id}_m + \Omega(\epsilon, q, v)}^{-1}\paren{v - \frac{\epsilon}{2} \mathbf{G}(q)^{-1}\nabla U(q)} \\
    \iff& \breve{v} + \Omega(\epsilon, q, v) \breve{v} = v - \frac{\epsilon}{2} \mathbf{G}(q)^{-1} \nabla U(q) \\
    \iff & \breve{v} + \frac{\epsilon}{2} \mathbf{G}(q)^{-1} \nabla U(q) = v - \Omega(\epsilon, q, \breve{v}) v \\
    \iff &  \breve{v} + \frac{\epsilon}{2} \mathbf{G}(q)^{-1} \nabla U(q) = \paren{\mathrm{Id}_m - \Omega(\epsilon, q, \breve{v})} v \\
    \iff & v = \paren{\mathrm{Id}_m - \Omega(\epsilon, q, \breve{v})}^{-1}\paren{\breve{v} + \frac{\epsilon}{2} \mathbf{G}^{-1}(q)\nabla U(q)}.
\end{align}
An identical series of computations reveals,
\begin{align}
    \breve{v} = \paren{\mathrm{Id}_m - \Omega(\epsilon, \tilde{q}, \tilde{v})}^{-1}\paren{\tilde{v} + \frac{\epsilon}{2} \mathbf{G}^{-1}(\tilde{q})\nabla U(\tilde{q})}.
\end{align}
Hence, applying $\Phi_{-\epsilon}$ to $(\tilde{q}, \tilde{v})$ yields the following series of updates,
\begin{align}
    \label{lagrangian-remarks:eq:negated-i} \breve{v}' &= \paren{\mathrm{Id}_m - \Omega(\epsilon, \tilde{q}, \tilde{v})}^{-1}\paren{\tilde{v} + \frac{\epsilon}{2} \mathbf{G}^{-1}(\tilde{q})\nabla U(\tilde{q})} \\
    &= \breve{v} \\
    q' &= \tilde{q} - \epsilon \breve{v}' \\
    &= \tilde{q} - \epsilon \breve{v} \\
    &= q \\
    \label{lagrangian-remarks:eq:negated-ii} v' &= \paren{\mathrm{Id}_m - \Omega(\epsilon, q, \breve{v})}^{-1}\paren{\breve{v} + \frac{\epsilon}{2} \mathbf{G}^{-1}(q)\nabla U(q)} \\
    &= v.
\end{align}
Hence we return to the initial condition $(q, v)$. This verifies that the explicit integrator employed in LMC is self-adjoint.
\end{proof}

\begin{proof}[Proof of \Cref{lagrangian-remarks:prop:composition-order}]
Since $\hat{\Phi}_\epsilon$ is $p$-th order accurate for $\Phi_\epsilon$ we have, by Taylor series expansion, that,
\begin{align}
    \label{eq:2021-12-19-taylor-i} \hat{\Phi}_0 &= \Phi_0 \\
    \label{eq:2021-12-19-taylor-ii} \frac{\mathrm{d}^k}{\mathrm{d}\epsilon^k} \hat{\Phi}_\epsilon(z_0) \bigg|_{\epsilon=0} &= \frac{\mathrm{d}^k}{\mathrm{d}\epsilon^k} \Phi_\epsilon(z_0) \bigg|_{\epsilon=0},
\end{align}
for $k=1,\ldots, p$. Let $f : \R\to\R^m$, then the Taylor series of expansion of $\Xi\circ f(\epsilon)$ is,
\begin{align}
    \label{eq:2021-12-19-taylor-composition} \Xi \circ f(\epsilon) &= \Xi\circ f(0) + \sum_{k=1}^p \epsilon^k g^{[k]}\paren{f(0), \frac{\mathrm{d}}{\mathrm{d}\epsilon} f(0), \ldots, \frac{\mathrm{d}^k}{\mathrm{d}\epsilon^k} f(0)} + \mathcal{O}(\epsilon^{p+1}),
\end{align}
where $g^{[k]}$ are functions determining the Taylor series coefficients which depend on the derivatives of $f$. For instance,
\begin{align}
    g^{[1]}\paren{f(0), \frac{\mathrm{d}}{\mathrm{d}\epsilon} f(0)} &= \nabla \Xi(f(0)) \cdot \frac{\mathrm{d}}{\mathrm{d}\epsilon} f(0) \\
    g^{[2]}\paren{f(0), \frac{\mathrm{d}}{\mathrm{d}\epsilon} f(0), \frac{\mathrm{d}^2}{\mathrm{d}\epsilon^2} f(0)} &= \nabla^2 \Xi(f(0))\paren{\frac{\mathrm{d}}{\mathrm{d}\epsilon} f(0), \frac{\mathrm{d}}{\mathrm{d}\epsilon} f(0)} + \nabla\Xi(f(0)) \cdot \frac{\mathrm{d}^2}{\mathrm{d}\epsilon^2} f(0).
\end{align}
By the equality of the Taylor series expansion coefficients in \cref{eq:2021-12-19-taylor-i,eq:2021-12-19-taylor-ii} it follows from \cref{eq:2021-12-19-taylor-composition} that,
\begin{align}
    \Xi\circ\Phi_\epsilon(z_0) - \Xi\circ\hat{\Phi}_\epsilon(z_0) = \mathcal{O}(\epsilon^{p+1}).
\end{align}
This proves that $\Xi\circ\hat{\Phi}_\epsilon$ is also $p$-th order accurate for $\Xi\circ\Phi_\epsilon$.
\end{proof}
\clearpage
\section{Propagator Matrices for the Leapfrog and Inverted Leapfrog}\label{lagrangian-remarks:app:propagator-matrices}

\begin{definition}\label{lagrangian-remarks:def:standard-leapfrog}
  The {\it generalized leapfrog integrator} for the Hamiltonian equations of motion in \cref{lagrangian-remarks:eq:hamiltonian-position,lagrangian-remarks:eq:hamiltonian-momentum} with Hamiltonian $H(q, p) = U(q) + p^\top p / 2$ is a map $(q, p)\mapsto (\tilde{q}, \tilde{p})$ defined by,
  \begin{align}
  \breve{p} &= p - \frac{\epsilon}{2} \nabla U(q) \\
  \tilde{q} &= q + \epsilon ~\breve{p} \\
  \tilde{p} &= \breve{p} - \frac{\epsilon}{2} \nabla U(\tilde{q}).
  \end{align}
\end{definition}
\begin{definition}\label{lagrangian-remarks:def:inverted-leapfrog}
  The {\it generalized leapfrog integrator} for the Hamiltonian equations of motion in \cref{lagrangian-remarks:eq:hamiltonian-position,lagrangian-remarks:eq:hamiltonian-momentum} with Hamiltonian $H(q, p) = U(q) + p^\top p / 2$ is a map $(q, p)\mapsto (\tilde{q}, \tilde{p})$ defined by,
  \begin{align}
  \breve{q} &= q + \frac{\epsilon}{2} p \\
  \tilde{p} &= p - \epsilon \nabla U(q) \\
  \tilde{q} &= \breve{q} + \frac{\epsilon}{2} \tilde{p}.
  \end{align}
\end{definition}

For Hamiltonians of the form $H(q, p) = \omega^2 q^2 / 2 + p^2 / 2$, the action of the leapfrog and inverted leapfrog integrators are linear. This means that there are matrices, called ``propagator matrices,'' which, when acting on the vector $(q, p)\in\R^2$, produce the same position in phase space as the integrators themselves. Computing integer matrix powers of these matrices can then produce the multi-step output of the integrators. In the case of the leapfrog integrator, the propagator matrix is \citep{leimkuhler_reich_2005},
\begin{align}
  \begin{pmatrix} \tilde{q} \\ \tilde{p} \end{pmatrix} = \underbrace{\begin{pmatrix} 1 - \frac{\epsilon^2\omega^2}{2} & \epsilon \\ \epsilon\omega^2\paren{1 - \frac{\epsilon^2\omega^2}{4}} & 1 - \frac{\epsilon^2\omega^2}{2} \end{pmatrix}}_{\mathbf{R}'} \begin{pmatrix} q \\ p \end{pmatrix}.
\end{align}
The propagator matrix for the inverted leapfrog is,
\begin{align}
  \begin{pmatrix} \tilde{q} \\ \tilde{p} \end{pmatrix} = \underbrace{\begin{pmatrix} 1 - \frac{\epsilon^2\omega^2}{2} & \epsilon\paren{1 - \frac{\epsilon^2\omega^2}{4}} \\ -\epsilon\omega^2 & 1 - \frac{\epsilon^2\omega^2}{2} \end{pmatrix}}_{\mathbf{R}''} \begin{pmatrix} q \\ p \end{pmatrix}.
\end{align}
At stationarity,
\begin{align}
  \begin{pmatrix} q \\ p \end{pmatrix} \sim \mathrm{Normal}\paren{\begin{pmatrix} 0 \\ 0 \end{pmatrix}, \begin{pmatrix} 1 / \omega^2 & 0 \\ 0 & 1 \end{pmatrix}}.
\end{align}
Therefore, for an integrator with propagator matrix $\mathbf{R}$, the distribution of the $k$-step transition is,
\begin{align}
  \mathrm{Normal}\paren{\begin{pmatrix} 0 \\ 0 \end{pmatrix}, \mathbf{R}^k\begin{pmatrix} 1 / \omega^2 & 0 \\ 0 & 1 \end{pmatrix} (\mathbf{R}^k)^\top},
\end{align}

\clearpage
\section{Involutive Monte Carlo}\label{lagrangian-remarks:app:involutive-monte-carlo}

\begin{proof}[Proof of \Cref{lagrangian-remarks:prop:involutive-monte-carlo-detailed-balance}]
The detailed balance condition states that for any Borel sets $(A, B), (C, D) \in \mathfrak{B}(\R^m\times\R^m)$ we have,
\begin{align}
    \int_{(A, B)} K((q, p), (C, D)) ~\pi(q, p)~\mathrm{d}q~\mathrm{d}p =\int_{(C, D)} K((q, p), (A, B)) ~\pi(q, p)~\mathrm{d}q~\mathrm{d}p.
\end{align}
Let $\pi(q, p)$ be the probability density on $\R^m\times\R^m$ that is proportional to $\exp(-H(q, p))$. Using the fact that $\Phi$ is an involution, it suffices to verify
\begin{align}
    &\int_{(A, B)} \alpha((q, p), \Phi(q, p), \abs{\mathrm{det}\paren{\nabla \Phi(x)}}) \mathbf{1}\set{(\tilde{q},\tilde{p})\in (C, D)} ~\pi(q, p)~\mathrm{d}q~\mathrm{d}p \\
    =& \int_{(A, B)} \min\set{1, \frac{\pi(\Phi(q, p))}{\pi(q, p)} \abs{\mathrm{det}\paren{\nabla \Phi(x)}}} \mathbf{1}\set{\Phi(q, p)\in (C, D)} ~\pi(q, p)~\mathrm{d}q~\mathrm{d}p \\
    =& \int_{(A, B)} \min\set{\pi(q, p), \pi(\Phi(q, p)) \abs{\mathrm{det}\paren{\nabla \Phi(x)}}} \mathbf{1}\set{\Phi(q, p)\in (C, D)} ~\mathrm{d}q~\mathrm{d}p \\
    =& \int_{(A, B)} \min\set{\frac{\pi(q, p)}{\pi(\Phi(q, p)) \abs{\mathrm{det}\paren{\nabla \Phi(x)}}}, 1} \mathbf{1}\set{\Phi(q, p)\in (C, D)} \pi(\Phi(q, p)) \abs{\mathrm{det}\paren{\nabla \Phi(x)}}~\mathrm{d}q~\mathrm{d}p \\
    =& \int_{(A, B)} \min\set{\frac{\pi(q, p)}{\pi(\Phi(q, p))} \abs{\mathrm{det}\paren{\nabla \Phi(\Phi(x))}}, 1} \mathbf{1}\set{\Phi(q, p)\in (C, D)} \pi(\Phi(q, p)) \abs{\mathrm{det}\paren{\nabla \Phi(x)}}~\mathrm{d}q~\mathrm{d}p \\
    =&  \int_{\Phi(A, B)} \min\set{1, \frac{\pi(\Phi(\tilde{q}, \tilde{p}))}{\pi(\tilde{q}, \tilde{p})} \abs{\mathrm{det}\paren{\nabla \Phi(\tilde{q}, \tilde{p})}}} \mathbf{1}\set{(\tilde{q},\tilde{p})\in (C, D)} ~\pi(\tilde{q}, \tilde{p}) ~\mathrm{d}\tilde{q}~\mathrm{d}\tilde{p} \\
    =& \int_{(C, D)} \alpha((q, p), \Phi(q, p),  \abs{\mathrm{det}\paren{\nabla \Phi(q, p)}}) \mathbf{1}\set{(q, p)\in \Phi(A, B)} ~\pi(q, p)~\mathrm{d}q~\mathrm{d}p \\
    =& \int_{(C, D)} \alpha((q, p), \Phi(q, p),  \abs{\mathrm{det}\paren{\nabla \Phi(q, p)}}) \mathbf{1}\set{\Phi(q, p)\in (A, B)} ~\pi(q, p)~\mathrm{d}q~\mathrm{d}p
\end{align}
Moreover, integrating over the rejection components of the transition kernel already have symmetry in $(A, B)$ and $(C, D)$:
\begin{align}
    &\int_{(A, B)} (1- \alpha((q, p), (\tilde{q}, \tilde{p}), \abs{\mathrm{det}\paren{\nabla \Phi(x)}})) \mathbf{1}\set{(q, p)\in (C, D)} ~\pi(q, p)~\mathrm{d}q~\mathrm{d}p \\
    =& \int_{(C, D)} (1- \alpha((q, p), (\tilde{q}, \tilde{p}), \abs{\mathrm{det}\paren{\nabla \Phi(x)}})) \mathbf{1}\set{(q, p)\in (A, B)} ~\pi(q, p)~\mathrm{d}q~\mathrm{d}p.
\end{align}
Detailed balance follows as a consequence.
\end{proof}

\begin{lemma}
Let $\Phi_\epsilon : \R^m\times\R^m\to\R^m\times\R^m$ be the Lagrangian leapfrog integrator with step-size $\epsilon\in \R$. Then $\mathbf{F}\circ \Phi_\epsilon$ is an involution.
\end{lemma}
\begin{proof}
Let $(\tilde{q},\tilde{v}) = \Phi_\epsilon(q, v)$. Now we apply the Lagrangian leapfrog to $(\tilde{q},-\tilde{v})$.
\begin{align}
    \breve{v}'' &= (\mathrm{Id}_m + \Omega(\epsilon, \tilde{q}, -\tilde{v}))^{-1} \paren{-\tilde{v} - \frac{\epsilon}{2} \mathbf{G}^{-1}(\tilde{q}) \nabla U(\tilde{q})} \\
    &= -(\mathrm{Id}_m - \Omega(\epsilon, \tilde{q}, \tilde{v}))^{-1} \paren{\tilde{v} + \frac{\epsilon}{2} \mathbf{G}^{-1}(\tilde{q}) \nabla U(\tilde{q})} \\
    &= -\breve{v}
\end{align}
from \cref{lagrangian-remarks:eq:negated-i}. Moreover,
\begin{align}
    \tilde{q}'' &= \tilde{q} + \epsilon \breve{v}'' \\
    &= \tilde{q} - \epsilon \breve{v} \\
    &= q.
\end{align}
Finally,
\begin{align}
    \tilde{v}'' &= (\mathrm{Id}_m + \Omega(\epsilon, q, -\breve{v}))^{-1} \paren{-\breve{v} - \frac{\epsilon}{2} \mathbf{G}^{-1}(q) \nabla U(q)} \\
    &= -(\mathrm{Id}_m - \Omega(\epsilon, q, \breve{v}))^{-1} \paren{\breve{v} + \frac{\epsilon}{2} \mathbf{G}^{-1}(q) \nabla U(q)} \\
    &= -v
\end{align}
from \cref{lagrangian-remarks:eq:negated-ii}.
\end{proof}

\bibliography{thebib}

\end{document}